\documentclass{LMCS}

\def\doi{8(3:19)2012}
\lmcsheading%
{\doi}
{1--39}
{}
{}
{Oct.~\phantom02, 2010}
{Sep.~19, 2012}
{}
 
\usepackage{pdfsync}
\usepackage{enumerate}
\usepackage{hyperref}
\usepackage{amsmath}
\usepackage{logiclncs-thm}
\usepackage{ifthen}
\usepackage[svgnames]{xcolor}
\usepackage{graphicx}
\usepackage{mathrsfs}

\usepackage{amssymb}
\usepackage[normalem]{ulem}

\newcommand{\addend}{\color{black}}

\newcommand{\impaddbeg}{\color{Red}}
\newcommand{\impadded}[1]{\impaddbeg #1 \addend}

\theoremstyle{plain}

\newcommand{\hole}{\Box}
\newcommand{\templ}[1] {\mathsf{TL}[#1]}
\newcommand{\langclass}{\mathscr L}
\newcommand{\algclass}{\mathscr A}
\newcommand{\tmodels}{\models_t}
\newcommand{\fmodels}{\models_f}
\newcommand{\tle}{\mathsf E}
\newcommand{\tlx}{\mathsf X}
\newcommand{\tlf}{\mathsf F}
\newcommand{\tlu}{\mathsf U}

\begin{document}

\title[Wreath Products of Forest Algebras, with  Applications to Tree Logics]{Wreath Products of Forest Algebras, with  Applications to Tree Logics\rsuper*}

\author[M.~Boja\'nczyk]{Mikolaj Boja\'nczyk\rsuper a}
\address{{\lsuper a}University of Warsaw}
\email{bojan@mimuw.edu.pl}

\author[H.~Straubing]{Howard Straubing\rsuper b}
\address{{\lsuper b}Boston College}	
\email{Straubing@cs.bc.edu}
\thanks{{\lsuper b}Research supported by  NSF Grant
    CCF-0915065}	

\author[I.~Walukiewicz]{Igor Walukiewicz\rsuper c}	
\address{{\lsuper c}LaBRI (Universit\'e de Bordeaux - CNRS)}
\email{igw@labri.fr}	
\thanks{{\lsuper c}Research supported by ANR 2010 BLAN 0202 01 FREC}	

\titlecomment{{\lsuper*}Extended abstract of the paper has appeared at LICS'09}

\begin{abstract} 
  We use the recently developed theory of forest algebras to find
  algebraic characterizations of the languages of unranked trees and
  forests definable in various logics. These include the temporal
  logics {CTL} and {EF}, and first-order logic over the ancestor
  relation.  While the characterizations are in general non-effective,
  we are able to use them to formulate necessary conditions for
  definability and provide new proofs that a number of languages are
  not definable in these logics.
\end{abstract}

\maketitle

\section{Introduction}

Logics for specifying properties of labeled trees play an important
role in several areas of Computer Science.  We say that a class of
regular languages of trees $\langclass$ has an \emph{effective
  characterization} if there is an algorithm which decides if a given
regular language of trees belongs to $\langclass$.  Effective
characterizations are known only for a few logics. In particular, we
do not know if such characterizations exist for the classes of
languages defined by the most common logics such as : CTL, CTL*, PDL,
or first-order logic with the ancestor relation.

In this paper we consider logics for {\it unranked} trees, in which
there is no {\it a priori} bound on the number of children a node may
have. Many such logics, including all the logics that are considered
in this paper, are no more expressive than monadic second-order logic,
and thus the properties they define can be described using automata.
Barcelo and Libkin~\cite{barcelo-libkin} and Libkin~\cite{libkin}
catalogue a number of such logics and contrast their expressive power.
We use recently developed theory of forest algebras to find algebraic
characterisations of the languages of unranked trees definable in some
most common logics. While the characterizations are in general
non-effective, we are able to use them to formulate necessary
conditions for definability and provide new proofs that a number of
languages are not definable in these logics.

For properties of {\it words,} such questions
 have been fruitfully studied by algebraic means.
 Whether or not a regular word language $L$ can be defined in a given logic can often be determined by verifying some property of the {\it syntactic monoid} of $L$---the transition monoid of the minimal automaton of $L.$  The earliest work in this direction is due to McNaughton and Papert~\cite{mcnaughton-papert} who studied first-order logic with linear order, and showed that a language is definable in this logic if and only if its syntactic monoid is aperiodic---that is, contains no nontrivial groups.  A comprehensive survey treating many different predicate logics is given in Straubing~\cite{straubing}; temporal logics are studied by Cohen, Perrin and Pin~\cite{perrin-pin} and Wilke~\cite{wilke}, among others.

 Algebraic techniques provide a striking alternative to purely
 model-theoretic methods for studying the expressive power of logics
 over words.  In many cases they have led to effective
 characterizations of certain logics, and actually to
 reasonably efficient algorithms. Even in the absence of effective
 characterizations, it is frequently possible to obtain effective
 necessary conditions for expressibility in a logic and use these to
 show the non-expressibility of certain languages.  For instance, the
 strictness of the $\Sigma_k$-hierarchy in first-order logic on
 words--the {\it dot-depth} hierarchy-- was first proved by such
 algebraic means (Brzozowski and Knast~\cite{brz-kna},
 Straubing~\cite{tcs-schutz}), while effective characterization of the
 levels of the hierarchy remains an open problem.

 There have been a number of efforts to extend this algebraic theory
 to trees; a notable recent instance is in the work of \'Esik and Weil
 on preclones~\cite{esik-weil1,esik-weil2}.
 Recently, Boja\'nczyk and Walukiewicz~\cite{bw}
 introduced  {\it forest
   algebras}, and along with it the {\it syntactic forest algebra,}
 which generalize monoids and  the syntactic monoid for
 languages of forests of unranked trees.  This algebraic model is
 rather simple, and in contrast to others studied in the literature,
 has already yielded effective criteria for definability in a number
 of logics: see Boja\'nczyk~\cite{boj2way},
 Boja\'nczyk-Segoufin-Straubing~\cite{boj-seg-str},
 Boja\'nczyk-Segoufin~\cite{boj-seg}.  Forest algebras are also
 implicit in the work of Benedikt and Segoufin~\cite{ben-seg} on
 first-order logic with successor and of Place and
 Segoufin~\cite{pl-seg} on locally testable tree languages.

 In the present paper we continue the study of forest
 algebras, by developing a theory of composition of forest
 algebras, using the {\it wreath product}.  The wreath
 product of transformation monoids plays an
 important role in the theory for words. In particular, it is
   connected to a composition operation on languages and to
   generalized temporal operators. This paper is concerned with
   describing the connection between formula composition and the
   wreath product of forest algebras, in the case of unranked
   trees. Here is a brief summary of our results:

\begin{enumerate}[(1)]

\item To each logic $\langclass$ among EF, CTL,  CTL*, first-order logic with ancestor, PDL and graded PDL, we associate a class of forest algebras, called the base of $\langclass$. We show that a language of forests is definable in the logic $\langclass$ if and only if it is recognized by an iterated wreath product of the forest algebras from the base of $\langclass$. (Theorem ~\ref{thm.main}.)

\item In the cases of EF and CTL, the base has a single forest algebra. For  the other cases we show that there is no finite base.  As a consequence, none of these logics can be generated by a finite collection of generalized temporal operators.  Using our algebraic framework, we give a simple and general proof of this fact. (Theorem ~\ref{thm.nofinitebase}.)

\item For the logics that do not have a finite base, we give an
  effective characterization of the base.  (Theorems
  ~\ref{thm:distributive-algebra} and ~\ref{thm:path-algebra}.) Note
  that an effective characterization of a base does not imply an
  effective characterization for wreath products of the base, so this
  result does not give an effective
  characterization of any of the logics mentioned in item~(1).

\item Going one step further, we  provide an effective
 characterization for the \emph{path languages}
 (Theorem~\ref{thm:path-algebra}): boolean combinations of languages
 from the base of graded PDL.

\item We give a new proof, based on the wreath product, of an effective characterization of the logic  {EF}. This result was proved earlier by other means. (Bojanczyk and Walukiewicz~\cite{bw}.)  Our argument here computes a decomposition based on the ideal structure of the underlying forest algebra.

\item Although we do not find effective
  characterizations for other prominent logics from our
    list, we are able to use our framework to
  establish necessary conditions for definability in these logics, and
  consequently to prove that a number of specific languages are not
  definable in them. (Theorem~\ref{thm:confusion}.)
   
\item We give an effective characterization of CTL* languages within
  first-order definable languages. Similarly for PDL languages within languages
  definable in graded PDL   (Theorem~\ref{thm.foctlstar}.)

\end{enumerate}  

\noindent{\it Plan of the paper.} In Sections 2-4 we present the basic terminology concerning, respectively, trees, logic, and forest algebras.  Our treatment of temporal logics is somewhat unorthodox, since our algebraic theory requires us to interpret formulas in forests as well as in trees, therefore the precise syntax and semantics are different in the two cases.  Section 4 includes a detailed treatment of the wreath product of forest algebras.  In Section 5 we establish the first of our main results, giving wreath product characterizations of all the logics under consideration.  In Section 6 we give the effective characterization of {EF}, and in Section 7 the necessary conditions for definability in the other logics.  Section 8 is devoted to applications of these conditions.

We note that \'Esik and Ivan~\cite{esik, esik-ivan} have done work of a similar flavor for { CTL} (for trees of bounded rank).  Our work here is of considerably larger scope, both in the number of different logics considered, and the concrete consequences our algebraic theory permits us to deduce. 

The present article is the complete version of an extended abstract presented at the 2009 IEEE Symposium on Logic in Computer Science.



\section{Trees, Forests and Contexts}

Let $A$ be a finite alphabet.  Formally, forests and trees over $A$
are expressions generated by the following rules: {\it (i)} if $s$ is
a forest and $a\in A$ then $as$ is a tree; {\it (ii)} if
$(t_1,\ldots,t_k)$ is a finite sequence of trees, then
$t_1+\cdots+t_k$ is a forest. We permit this summation to take place
over an empty sequence, yielding the {\it empty forest,} which we
denote by 0, and which gets the recursion started.  So, for example, the following forest  with two roots
\medskip
\begin{center}
	\includegraphics{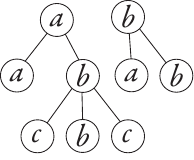}
\end{center}
is described by the  expression
\begin{eqnarray*}
	a(a0+b(c0+b0+c0))+b(a0+b0).
\end{eqnarray*}
Normally, when we write
such expressions, we delete the zeros.
 We denote
the set of forests over $A$ by $H_A.$ This set forms a monoid with
respect to forest concatenation $s+t$, with the empty forest $0$ being
the identity. We denote the set of trees over $A$ by $T_A.$

If $x$ is a node in a forest, then the {\it subtree} of $x$ is simply
the tree rooted at $x,$ and the {\it subforest} of $x$ is the forest
consisting of all subtrees of the children of $x.$ In other words, if
the subtree of $x$ is $as,$ with $a\in A$ and $s\in H_A,$ then the
subforest of $x$ is $s.$ Note that the subforest of $x$ does not
include the node $x$ itself, and is empty if $x$ is a leaf.

A {\it forest language} over $A$ is any subset of $H_A$.

A {\it context} $p$ over $A$ is formed by replacing a leaf of  a
nonempty forest by a special symbol $\hole$.  Think of $\hole$ as a kind
of place-holder, or {\it hole}.  Given a context $p$ and a forest $s$,
we form a forest $ps$ upon substituting $s$ for the hole in $p.$ In
the interpretation of forests as expressions, this really is just
substitution of the expression $s$ for the hole of $p$; the graphical
interpretation of this operation is depicted below.

\begin{center}
  \includegraphics[scale=0.7]{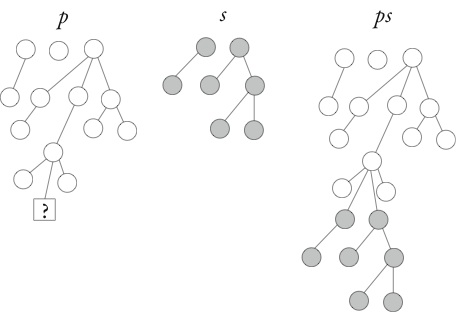}
\end{center}

In a similar manner, we can substitute another context $q$ for the
hole, and obtain a new context $pq.$ We obtain in this way a
composition operation on contexts. We denote the set of contexts over
$A$ by $V_A.$ This set forms a monoid, with respect to this
composition operation, with the empty context $\hole$ as the identity.

Note that for any $s,t \in H_A,$ $V_A$ contains a context $s+\hole+t,$
in which the hole has no parent, such that $(s+\hole+t)u=s+u+t$ for
all $u\in H_A.$

Our trees, forests and contexts are ordered, so
that $s+t$ is a different forest from $t+s$ unless $s=t$ or one of $s, t$
is 0.  This noncommutativity is important in a number of applications.
However the present article really deals with unordered trees, so there is no harm in thinking of $+$ as a
commutative operation on forests.




\section{Logics for Forest Languages}\label{section.logic}

We can define {\it regular} forest languages by means of an automaton
model that is a minor modification of the standard bottom-up tree
automaton.  The transition function has to be altered to cope with
unbounded branching, and the acceptance condition needs to take
account of the sequence of states in the roots of all the trees in the
forest.  See~\cite{bw} for a precise definition of such an automaton
model.  The usual equivalence between monadic second-order logic and
regularity holds in this setting.

For a general treatment of predicate and temporal logics for unranked
trees, we refer the reader to Libkin~\cite{libkin} and
Barcel\'o-Libkin~\cite{barcelo-libkin}. We will have to give a
somewhat different description of similar logics in order to express
properties of forests as well as of trees.  In all cases the logics
that we describe are fragments of monadic second-order logic, and thus
the languages they define are all regular forest languages.

\subsection{First-order logic for trees and forests}\label{subsection.fologic}
Let $A$ be a finite alphabet.  Consider first-order logic equipped
with unary predicates $Q_a$ for each $a\in A,$ and a single binary
predicate $\prec.$ Variables are interpreted as nodes in forests over
$A.$ Formula $Q_ax$ is interpreted to mean that node $x$ is labeled
$a,$ and $x\prec y$ to mean that node $x$ is
a (non-strict) ancestor of node $y.$ A {\it sentence} $\phi$---that
is, a formula without free variables---consequently defines a language
$L_{\phi}\subseteq H_A$ consisting of forests over $A$ that satisfy
$\phi.$ For example, the sentence
$$\exists x \exists y(Q_ax\wedge Q_ay\wedge \neg(x\prec y)\wedge\neg(y\prec x))$$
defines the set of forests containing two incomparable occurrences of $a.$  We denote this logic  by $FO[\prec].$   Note that this logic has no predicates to access the order of siblings. In particular, any language defined by the logic will be horizontally commutative, i.e.~closed under reordering sibling trees. 

It is more traditional to consider logics over trees rather than over forests.  For $FO[\prec]$ we need not worry too much about this distinction, since we can express in first-order logic the property that a forest has exactly one root (by the sentence $\exists x\forall y(x\prec y)$).  Thus the question of whether a given set of trees is first-order definable does not depend on whether we choose to interpret sentences in trees or in forests. 

\subsection{Temporal
  logics}~\label{subsection.temporallogic} We describe here a general
framework for temporal logics interpreted in trees and forests.  By
setting appropriate parameters in the framework we generate all sorts
of 
temporal logics that are traditionally studied.

The general framework is  called \emph{graded propositional
  dynamic logic} (graded PDL).

\noindent{\it Syntax of temporal formulas.} We distinguish between two kinds of formulas: {\it tree formulas} and {\it forest formulas}.  The syntax of these formulas is defined by mutual recursion, as follows:

\begin{iteMize}{$\bullet$}
\item {\bf T} and {\bf F} are forest formulas.
\item If $a\in A,$ then $a$ is a tree formula.  (Such formulas are called \emph{label formulas.})
\item Finite boolean combinations of tree formulas are tree formulas, and finite boolean combinations of forest formulas are forest formulas.
\item  Every forest formula is a tree formula.
\item Before defining the key construction we need to introduce the
  concept of an \emph{unambiguous} set of formulas. Such a set
  $\Phi=\set{\phi_1,\dots,\phi_{n+1}}$ is 
  constructed from a sequence of
  tree formulas $\psi_1,\ldots,\psi_n$ by a simple syntactic
  operation ensuring that every tree satisfies exactly one
  formula from $\Phi$:
	\begin{eqnarray*}
          \phi_i=
          \psi_i \land \bigwedge_{j \le {i-1}} \neg \psi_j  \mbox{ for
            $i=1,\ldots,n$} \qquad\mbox{ and }\qquad \phi_{n+1}= 
	\bigwedge_{j \le n} \neg \psi_j.
	\end{eqnarray*}

\item If  $\Phi$ is a finite unambiguous set of tree formulas, $k>0$
  is an integer,  and
  $L\incl \Phi^*$ is a 
  regular language then  $\tle^k L$ is a forest formula.
\end{iteMize}

\noindent{\it Semantics of temporal formulas.}   We define two notions of satisfaction: 
\emph{tree satisfaction} $t \tmodels \phi$, where $t$ is a tree and  $\phi$ is a tree formula, which
coincides with the usual notion of satisfaction; and  \emph{forest
  satisfaction} $t \fmodels \varphi$, where $t$ is a forest and $\phi$ is a forest formula, which is somewhat unusual. Again, these relations are defined by mutual recursion.
  \begin{iteMize}{$\bullet$}
  \item If $t\in H_A$ then $t\fmodels {\bf T}$ and $t\not\fmodels {\bf F}.$
  \item If $t\in T_A$ and $a\in A,$ then $t\tmodels a$ if and only if the root node of $t$ is labeled $a.$
  \item Boolean operations have their usual meaning; {\it e.g.,} if $t\in H_A$ and $\phi_1,$ $\phi_2$ are forest formulas, then $t \fmodels\phi_1\wedge\phi_2$ if and only if $t\fmodels \phi_1$ and $t\fmodels\phi_2.$
  \item Let $\phi$ be a forest formula and $t\in T_A,$  so that $t=as$ for a unique $s\in H_A,$ 
  $a\in A.$ Then $t\tmodels\phi$ if and only if $s\fmodels\phi.$

\item Let $k > 0$, and let $\Phi$ be a finite unambiguous set of
    tree formulas, with $L\subseteq\Phi^*$ a regular language. Let
    $s\in H_A.$ If $x$ is a node of $s,$ then we label the node by
    $\phi_i$ if $t_x\tmodels\phi_i,$ where $t_x$ is the subtree of
    $x.$ Note that because of the unambiguity requirement, there is
    exactly one such label. A path $x_1,\ldots, x_n$ of consecutive
    nodes of $s$ beginning at a root, but not necessarily extending to
    a leaf, thus yields a unique word
    $\phi_{i_1}\ldots\phi_{i_n}\in\Phi^*,$ which we call the
    $\Phi$-path of $x_n.$ We say $s\fmodels\tle^kL$ if there are at least
    $k$ nodes in the forest whose $\Phi$-path belongs to $L.$
  \end{iteMize}
\noindent We stress that in counting paths, we do not require the paths to be disjoint, and we do not require them to extend all the way to the leaves.  For example, the forest $aa+a$ contains three different nonempty paths from the root, so this forest satisfies the formula $\tle^3 a^+.$

Given a  temporal formula $\psi$, we write $L_\psi$ for the set of
forests that forest satisfy $\psi$:
\begin{equation*}
  L_\psi=\set{s\in H_a : s\fmodels \psi}
\end{equation*}

We specialize the above framework by restricting either the value of $k$ or the language $L$ in the application of the operator $E^kL,$ or both.  This leads, in the case of trees, to some logics that have been widely studied.  We catalogue these below:




\paragraph*{\it EF} As a first example, we show how to implement the
operator ``there exists a descendant'', often denoted by $\tle
\tlf$. This example also highlights the difference between tree
satisfaction and forest satisfaction.  Consider the special case of
$\tle^k L$ where, for some tree formula $\psi$,
\begin{equation}\label{eq:ef-syntactic-restriction}
	\Phi = \set{\psi,\neg \psi} \qquad k=1 \qquad L=(\neg \psi)^* \psi.
\end{equation}
In this special case, we write $\tle \tlf \psi$ instead of $\tle^k
L$. It is easy to see that $\tle \tlf \psi$ is forest satisfied by a
forest $s$ if and only if $\psi$ is satisfied by some subtree of
$s$. In the special case when $s$ is a tree, this subtree may be $s$
itself.  The semantics 
shows how to interpret $\tle \tlf \psi$ as a tree formula.  If $t$ is a
tree, then $t$ tree satisfies $\tle \tlf \psi$ if and only if $t$ has
a proper subtree that satisfies $\psi$. In other words, tree
satisfaction of $\tle \psi$ corresponds to the so called “strict
semantics”, while forest satisfaction of $\tle \psi$ corresponds to
the “non strict semantics”. We will use the term EF for the fragment
of graded PDL where the operator $\tle^k L$ is only used in the
special case of $\tle \tlf \psi$.

  \paragraph*{\it CTL } As a second example, we show how to implement the operator $\tle \psi \tlu \phi$ of CTL.
Consider the special case of $\tle^k L$ where, for some tree formulas $\psi$ and $\phi$,
\begin{equation}\label{eq:ctl-syntactic-restriction}
	\Phi = \set{\psi\wedge\neg\phi,\neg\psi\wedge\neg\phi,\phi} \qquad k=1 \qquad L=(\psi\wedge\neg\phi)^*\phi.
\end{equation}
In this special case, we write $\tle \psi \tlu \phi$ instead of
$\tle^k L$. It is easy to see that $\tle \psi \tlu \phi$ is forest
satisfied by a forest $s$ if and only if the subtree of some node $x$
tree satisfies the formula $\phi,$ and the subtree at every proper
ancestor of $x$ tree satisfies $\psi.$ Let us look now at
  the tree semantics of the formula $\tle \psi \tlu
\phi$.  If $t$ is a tree, then $t$ tree satisfies $\tle \psi \tlu
\phi$ if and only if the subtree of some non-root node $x$ tree
satisfies $\phi$, and every non-root proper ancestor of $x$ tree
satisfies $\psi$. As was the case with the operator $\tle\tlf$, tree satisfaction corresponds to “strict semantics” and forest satisfaction corresponds to “non strict semantics”. We will use
the term CTL for the fragment of graded PDL where the operator $\tle^k
L$ is only used in the special case of $\tle \psi \tlu
\phi$. \footnote{In most presentations CTL also has the ``next''
  operator $\tle \tlx \phi$, as well as the dual operator $\tle
  \neg(\psi \tlu \phi)$. The next operator is redundant thanks to the
  strict semantics, and the dual operator is redundant in finite
  trees.}

\paragraph*{\it First-order logic }
We  use our temporal framework to
  characterize the languages definable in $FO[\prec].$ 

 \begin{thm}\label{thm.fo} A forest language is definable in
     $FO[\prec]$ if and only if it is definable by a forest formula in
     which the operator $\tle^k L$ is restricted to word languages $L$
     that are first-order definable over an unambiguous finite alphabet
     $\Phi$ of tree formulas.
 \end{thm}

 \proof The theorem is very similar to the result of Hafer and
 Thomas~\cite{haferthomas} who show that first-order logic
 coincides with CTL* on finite binary trees. The theorem is even
 closer to the result  Moller and
 Rabinovich~\cite{moller-rabinovich} who show that over infinite
 unranked trees Counting-CTL* is
 equivalent to \emph{Monadic Path Logic (MPL)}. To deduce our theorem
 from their result it is enough to clarify the relations between
 different logics. 
 
 The logics considered by Moller and Rabinovich express properties of
 infinite unranked trees, which can have both infinite branches and
 finite branches that end in leaves.  A {\it maximal path} is
 therefore defined as a path that begins in any node, is directed away
 from the root, and either continues infinitely or ends in a leaf.
 Monadic Path Logic (MPL) is the restriction of monadic second-order logic over the
 predicate $\prec$ in which second-order quantification is restricted
 to maximal paths. In other words, MPL is the extension of $FO[\prec]$
 that allows quantification over maximal paths.  Over infinite trees,
 MPL is more expressive than first-order logic, since it can define
 the property ``some path contains infinitely many $a$'s '', which
 cannot be defined in $FO[\prec]$. However, over finite unranked
 trees, $FO[\prec]$ has the same expressive power as MPL. This is
 because a maximal path in a finite tree can be described by its first
 node and the leaf where it ends.

 The logic counting-CTL* can be interpreted as the the fragment
 of graded PDL where the operator $\tle^k L$ is only allowed in the
 following two restricted forms:
\begin{iteMize}{$\bullet$}
\item A next operator $\tlx^k$. This formula holds in a
  tree if subtrees of at least $k$ children of the root satisfy
  $\phi$. If $\tlx^k\phi$ is a formula of counting-CTL* and $\wh\phi$
  is a translation of $\phi$ into graded-PDL then $\tlx^k\phi$ is
  translated into a tree formula $\tle^kL_\phi$ where
  $L_\phi=\set{\wh\phi}$. Indeed, such a formula requires existence of
  $k$ different paths of length $1$ whose labellings belong to
  $\wh\phi$. 

\item An existential path operator, which we denote here by
    $\tle'$ 
  (the original paper uses $\tle$, but we use $\tle'$ to highlight the
  slight change in semantics). This operator works like our $\tle L$,
  but with the the difference that $\tle' L$ is a tree formula, and
  the path begins in the unique root of the tree. 
  Rabinovich and Moller require that $L$ is definable in LTL, which is
  equivalent to first-order definability.
\end{iteMize}
So counting-CTL* can be translated to a fragment of graded-PDL
  using only first-order definable word languages in
  quantification. Hence, by the result of Moller and Rabinovich we get
  a translation of $FO[\prec]$ to this fragment. The translation in the opposite
  direction is straightforward.  \qed

Note that Theorem~\ref{thm.fo} fails without the restriction on
unambiguity of the alphabet $\Phi$.  For instance, if we took
$A=\{a,b,c\},$ $\Phi=\{\phi_1,\phi_2 \},$ where $\phi_1=a\vee c,$
$\phi_2=b\vee c,$ then $L=(\phi_1\phi_2)^+$ is first-order definable
as a word language. One can imagine what the semantics of $\tle
  L$ should be in the case of such $\Phi$: a node labelled with $c$ can
  be labelled either with $\phi_1$ or with $\phi_2$.  With
  this semantics however, the language defined by $\tle L$ is not
first-order definable. (If it were, we would be able to define in
first-order logic the set of forests consisting of a single path with
an even number of occurrences of $c$.)

Actually, one can show, using composition theorems similar to
  those used by Hafer and Thomas, or Moller and
  Rabinovich, that graded PDL has the same expressive power as chain
  logic, which is the fragment of monadic second order logic where set
  quantification is restricted to chains, i.e.~subsets of paths.

\paragraph*{CTL* and PDL} 
Finally, we define two more temporal logics by modifying the
definitions above.  CTL* is like the fragment of temporal logic in
Theorem~\ref{thm.fo}, except that we only allow $k=1$ in $\tle^k L$.
In particular, CTL* is a subset of $FO[\prec].$ We also consider PDL,
which is obtained by restricting the temporal formulas $\tle^k L$ to
$k=1$, but without the requirement that  $L$ be first-order definable.
If we place no restriction on either the multiplicity $k$ or the
regular language $L,$ we obtain {\it graded PDL.}  

\subsection{Language composition and bases} 
\label{section.composition} 
In this section we provide a more general notion of temporal logic,
where the operators are given by regular forest languages.  This is
similar to notions introduced by \'Esik in \cite{esik}. The benefit of the general framework is twofold. First, it
corresponds nicely with the algebraic notion of wreath product
presented later in the paper.  Second, it allows us to state and prove
negative results, for instance our infinite base theorem, which says
that the number of operators needed to obtain first-order logic is
necessarily large.

We introduce a composition operation on forest
languages.  Fix an alphabet $A$, and let $\{L_1,\ldots,L_k\}$ be a
partition of $H_A$. Let $B=\{b_1,\ldots,b_k\}$ be
another alphabet, with one letter $b_i$ for each block $L_i$ of the
partition. The partition and alphabet are used to define a relabeling
\begin{eqnarray*}
  t \in H_A \qquad \mapsto \qquad t[L_1,\ldots,L_k] \in H_{A
    \times B}
\end{eqnarray*}
in the following manner. The  nodes in the forest $t[L_1,\ldots,L_k]$ are
the same as in the forest $t$, but the labels are different. A node $x$ that had label $a$ in $t$ gets label  $(a,b_i)$ in the new forest, where $b_i$ corresponds to the unique language $L_i$
that contains the subforest of $x$ in $t$. For the partition and $B$
as above, and $L$ a language of forests over $A \times B$, we define
$L[L_1,\ldots,L_k] \subseteq H_A$ to be the set of all forests $t$
over $A$ for which $t[L_1,\ldots,L_k]\in L.$

The operation of language composition is similar to formula
composition. The definitions below use this intuition, in order to
define a ``temporal logic'' based on operators given as forest
languages. Formally, we will define the closure of a language class
under language composition. First however, we need to comment on a
technical detail concerning alphabets.  In the discussion below, a
forest language is given by two pieces of information: the forests it
contains, and the input alphabet.  For instance, we distinguish
between the set $L_1$ of all forests over alphabet $\set a$, and the
set $L_2$ of all forests the alphabet $\set {a,b}$ where $b$ does not
appear. The idea is that sometimes it is relevant to consider a
language class $\langclass$ that contains $L_1$ but does not contain
$L_2$, such as the class of \emph{definite} languages that only look at a bounded prefix of the input forest  (such classes will not appear in this particular
paper). This distinction will be captured by our notion of language
class: a language class is actually a mapping $\langclass$, which
associates to each finite alphabet a class of languages over this
alphabet.

Let $\langclass$ be a class of forest languages, which will be called
the \emph{language base}. The temporal logic with language base
$\langclass$ is defined to be the smallest class $\templ \langclass$ of
forest languages that contains $\Ll$ and is closed under
boolean combinations and language
composition, i.e.
\begin{eqnarray*}
  L_1,\ldots,L_k, L \in  \templ \langclass \quad \Rightarrow \quad
  L[L_1,\ldots,L_k] \in \templ \langclass.
\end{eqnarray*}
Formally speaking, in the above we should highlight the alphabets (the languages $L_1,\ldots,L_k$ and $L[L_1,\ldots,L_k]$ belong to the part of $\templ \langclass$ for alphabet $A$, while the language $L$ belongs to the part of $\templ \langclass$ for alphabet $A \times B$, as in the definition of the composition operation).

We can translate  the definitions of the temporal logics we have considered in terms of
language composition.  This gives  the following theorem.
\begin{thm} \label{thm.composition}
The logics EF, CTL, $FO[\prec]$, CTL* , PDL and graded PDL have language bases as depicted
in Figure~\ref{fig:bases}.
\end{thm}

\begin{figure}
  \centering
  \begin{tabular}{l|l}
Logic & Languages in the language base for alphabet $A$\\
\hline
EF & $    \set{ \mbox{``some node with  $a$''}: a \in A}$ \\
CTL & $  \set{ \mbox{``some path in
        $B^*b$''}: B \subseteq A, b \in A}$\\
$FO[\prec]$ &    $ \set{ \mbox{``at least $k$ paths in $L$''}:  k \in
      \Nat,  L \in FO_A[<]}$\\
CTL*  &  $ \set{ \mbox{``some path in $L$''}:   L \in
  FO_A[<]}$\\
 PDL & $  \set{ \mbox{``some path in $L\subseteq A^+$''}:   L
      \mbox{ regular}}$\\
 graded PDL&$ \set{ \mbox{``at least $k$ paths in $L\subseteq A^+$''}:  k \in
      \Nat,  L  \mbox{ regular}}$
\end{tabular}
\caption{Language bases for temporal logics}
  \label{fig:bases}
\end{figure}

\noindent Note that the assertion about $FO[\prec]$ depends on Theorem
~\ref{thm.fo}.



\section{Forest Algebras}
\label{sec:forest-algebras}
\subsection{Definition of forest algebras}
Forest algebras, introduced in~\cite{bw} by Boja\'nczyk and Walukiewicz, extend the algebraic theory of syntactic monoid and
syntactic morphism for regular languages of words to the setting of
unranked trees and forests.  A {\it forest algebra} is a pair $(H,V)$
of monoids together with a faithful monoidal left action of $V$ on the
set $H.$ This means that for all $h\in H,$ $v\in V,$ there exists
$vh\in H$ such that (i) $(vw)h=v(wh)$ for all $v,w\in V$ and $h\in H,$
(ii) if $1\in V$ is the identity element, then $1h=h$ for all $h\in
H,$ and (iii) if $vh=v'h$ for all $h\in H,$ then $v=v'.$ We write the
operation in $H$ additively, and denote the identity of $H$ by 0. We
call $H$ and $V,$ respectively, the {\it horizontal} and {\it
  vertical} components of the forest algebra. The idea is that $H$
represents forests and $V$ represents contexts. As was the case with
the addition in $H_A,$ this is not meant to suggest that $H$ is a
commutative monoid, although in all the applications in the present
paper $H$ will indeed be commutative. We require one additional
condition: For each $h\in H$ there are elements $1+h, h+1\in V$ such
that for all $g\in H,$ $(1+h)g=g+h,$ and $(h+1)g=h+g.$ A consequence is that every element $h \in H$ can be written as $h=v0$ for some $v \in V$, namely $v=h+1$. A homomorphism
of forest algebras consists of a pair of monoid homomorphisms
$(\alpha_H,\alpha_V):(H,V)\to (H',V')$ such that
$\alpha_H(vh)=\alpha_V(v)\alpha_H(h)$ for all $v\in V$ and $h\in H.$
We usually drop the subscripts on the component morphisms and simply
write $\alpha$ for both these maps.

Of course, if $A$ is a finite alphabet, then $(H_A,V_A)$ is a forest
algebra.  The empty forest 0 is the identity of $H_A,$ and the empty
context $\hole$ is the identity of $V_A.$ This is the {\it free forest
  algebra} on $A,$ and we denote it $A^{\Delta}.$ It has the property
that if $(H,V)$ is any forest algebra and $f:A\to V$ is a map, then
there is a unique homomorphism $\alpha$ from $A^{\Delta}$ to $(H,V)$
such that $\alpha(a\hole)=f(a)$ for all $a\in A.$

\subsection{Recognition and  syntactic forest algebra}

Given a homomorphism $\alpha:A^{\Delta}\to (H,V),$ and a subset $X$ of
$H,$ we say that $\alpha$ {\it recognizes} the language
$L=\alpha^{-1}(X),$ and also that $(H,V)$ recognizes $L.$ A forest
language is regular if and only if it is recognized in this fashion by
a {\it finite} forest algebra.  Moreover, for every forest language
$L\subseteq H_A,$ there is a special homomorphism
$\alpha_L:(H_A,V_A)\to (H_L,V_L)$ recognizing $L$ that is minimal in
the sense that $\alpha_L$ is surjective, and factors through every
homomorphism that recognizes $L.$ 
We call $\alpha_L$ the {\it
  syntactic morphism} of $L,$ and $(H_L,V_L)$ the {\it syntactic
  forest algebra} of $L.$ If $s,s'\in H_A,$ then
$\alpha_L(s)=\alpha_L(s')$ if and only if for all $v\in V_A,$ $vh\in
L\Leftrightarrow vh'\in L.$ This equivalence is called the {\it
  syntactic congruence} of $L.$ An important fact in applications of
this theory is that one can effectively compute the syntactic morphism
and algebra of a regular forest language $L$ from any automaton that
recognizes $L.$ (See~\cite{bw}.)

We say that a forest algebra $(H_1,V_1)$ \emph{divides}
  $(H_2,V_2)$, in symbols $(H_1,V_1)\prec (H_2,V_2)$ if $(H_1,V_1)$ is
  a quotient of a subalgebra of $(H_2,V_2)$. In particular,
  $(H_L,V_L)$ {\it divides} every forest algebra that recognizes $L.$

There is a subtle point in the definition of division of forest
algebras given above that we will need to address.  We have defined
this in a way that directly generalizes the standard notion of
division of monoids: A divisor of a monoid $M$ is a quotient of a
submonoid of $M.$ But a forest algebra, is, in particular, a
transformation monoid, and there is a second notion of division, which
comes from the theory of transformation monoids, that will be
particularly useful when we deal with wreath products: We say that
$(H,V)$ {\it tm-divides} $(H',V')$ if there is a submonoid $K$ of
$H',$ and a surjective monoid homomorphism $\Psi:K\to H$ such that for
each $v\in V$ there exists $\hat v \in V'$ with $\hat v K\subseteq K,$
and for all $k\in K,$
$$\Psi(\hat vk)=v\Psi(k).$$
Fortunately, the two notions of division coincide, as shown in the following Lemma.

\begin{lem}\label{lemma.division}
Let $(H_1,V_1)$ and $(H_2,V_2)$ be forest algebras.  $(H_1,V_1)\prec(H_2,V_2)$ if and only if $(H_1,V_1)$ {\it tm}-divides $(H_2,V_2).$
\end{lem}
\proof First suppose $(H_1,V_1)$ divides $(H_2,V_2).$  Then there is a submonoid
$V'$ of $V_2$ and a  forest algebra homomorphism 
$$\alpha:(V'\cdot 0,V')\to (H_1,V_1).$$
(Strictly speaking, we should reduce $V'$ to the quotient that acts faithfully on $V'\cdot 0,$ but leaving this reduction out does not change the argument.)  Let $v\in V_1,$ and set $\hat v$ to be any 
element of $V'$ such that $\alpha({\hat v})=v.$  We then have for $h\in V'\cdot 0,$

$$\alpha({\hat v}h)=\alpha({\hat v})\alpha(h)=v\alpha(h),$$
so $(H_1,V_1)$ tm-divides $(H_2,V_2).$ 



Conversely, suppose  $(H_1,V_1)$ tm-divides $(H_2,V_2),$ with underlying homomorphism $\alpha:H'\to H_1.$   Let $A$ be an alphabet at least as large as $V_1,$ and let $\gamma:A\to V_1$
be an onto map.  This extends, because of the universal property of the free forest algebra, to a (surjective) forest algebra homomorphism $\gamma:A^{\Delta}\to(H_1,V_1).$  We define $\delta:A\to V_2$ by setting
$$\delta(a)=\widehat{\gamma(a)}$$
for all $a\in A,$ and consider its  extension $\delta$ to a forest algebra homomorphism.
It is enough to show that for $x,y\in V_{A},$ $\delta(x)=\delta(y)$ implies $\gamma(x)=\gamma(y).$  This will imply that $\gamma$ factors through $\delta$ and give the required division.

Observe that if $s\in H_{A},$ then $\delta(s)$ is in the domain $H'$ of $\alpha,$ because $s= x\cdot 0$ for some $x\in V_1,$ and thus  
\begin{eqnarray*}
\gamma(s) &=&\gamma(x)\gamma(0)\\
                    &=&\gamma(x)\alpha(\delta(0))\\
                    &=&\alpha(\widehat{\gamma(x)}\delta(0))\\
                    &=&\alpha(\delta(x)\delta(0))\\
                    &=&\alpha(\delta(x\cdot 0))\\
                    &=&\alpha(\delta(s)).
\end{eqnarray*} 
So by assumption, we have
$$\gamma(a)\alpha(\delta(s))=\alpha(\delta (a)\delta(s))$$
for all $s\in H_{A},$ $a\in A$.  A straightforward induction on the number of nodes in $x$ implies that for any $x\in V_{A},$
$$\gamma(x)\alpha(\delta(s))=\alpha(\delta(x)\delta(s)).$$
Now suppose $h\in H_1$ and $\delta(x)=\delta(y). $  As noted above, $h=\alpha(\delta(s))$ for some
$s\in H_{A},$ and consequently
\begin{eqnarray*}
\gamma(x)\cdot h &=&\gamma(x)\alpha(\delta(s))\\
	&=&\alpha(\delta(x)\delta(s))\\
	&=&\alpha(\delta(y)\delta(s))\\
	&=&\gamma(y)\alpha(\delta(s))\\
	&=&\gamma(y)\cdot h.
\end{eqnarray*}

Since $h$ was arbitrary, we get $\gamma(x)=\gamma(y),$ by faithfulness.

\qed

\subsection{Wreath product}

Here we introduce the wreath product of forest algebras. We first try to give some intuition behind the construction.  The wreath product originally arose in the theory of permutation groups, but it was subsequently adapted to provide an algebraic model of serial composition of automata.  The idea is that the first automaton reads an input word $a_1\cdots a_n$ beginning in state $q_0.$  The second automaton sees both the run of the first automaton on this input string, as well as the original input string---that is, it reads the sequence
$$(q_0,a_1), (q_0a_1,a_2),\ldots, (q_0a_1\cdots a_{n-1},a_n)$$
as an input word, beginning in its initial state $p_0.$    This defines a composite action of words over the original input alphabet $A$ on pairs of states $(p,q).$  The wreath product is, essentially, the transition monoid of this action.

The idea behind the wreath product of two forest algebras is also to 
model sequential composition. The first algebra `runs' on an input forest, and then a second automaton runs on the same forest, but also gets to see to see the run
of the first automaton.  We will make this composition precise by defining the
sequential composition of two homomorphisms. Assume that
\begin{eqnarray*}
   \alpha : A^\Delta \to (G,W)
\end{eqnarray*}
is a forest algebra homomorphism. For a forest $t$ over $A$, let
$t^\alpha$ be the forest over $A \times G$ obtained from $t$ by
changing the label of each node $x$ from $a$ to the pair $(a,g)$,
where $g\in G$ is the value assigned by $\alpha$ to the subforest of~$x$. In other words,
 $t^{\alpha}$ is the forest $t[L_1,\ldots,L_k],$ where $G=\{g_1,\ldots,g_k\}$ and $L_i=\alpha^{-1}(g_i).$ 
The sequential composition, will use a second homomorphism that reads
the relabeling $t^\alpha$ and yields a value in a second forest algebra; that is, 
\begin{eqnarray*}
  \beta : (A \times G)^\Delta \to (H,V)\ .
\end{eqnarray*}
The sequential composition of $\alpha$ and $\beta$ is the function
$\alpha \otimes \beta : H_A \to G \times H$ defined by
\begin{eqnarray*}
  t\quad \mapsto \quad  (\alpha(t),\beta(t^{\alpha}))\ .
\end{eqnarray*}
The wreath product $(G,W) \circ (H, V)$ of forest algebras is defined 
 to capture this notion of sequential composition.
While it is hardly surprising that there is an algebraic construction that
models sequential composition for forests, just as there is such a construction for words, it is rather remarkable that the construction for forest algebras is identical to the one used for transformation monoids. (In fact, one
could even argue that the wreath product is better suited to forest
languages, since it works directly on the forest algebra, while for
word languages one goes from monoids to transformation monoids.)


We now present the definition of the wreath product of two forest algebras $(H_1,V_1)$ and $(H_2,V_2)$. This wreath product denoted by $(H_1,V_1)\circ (H_2,V_2)$. 

Note that forest algebras are transformation monoids,  for which the wreath product is a classical operation. We will apply the classical definition without changes in this setting, yielding some of the ingredients of a forest algebra, namely: 1) the carriers of the horizontal and vertical monoid; 2)  the action of the vertical monoid on  the horizontal monoid; and 3) the composition operation in the vertical monoid. The missing ingredient, not given by the classical definition,  will be 4) the monoid operation in the horizontal monoid. 
  
We describe below the classical definition of wreath product of
transformation monoids, as applied to the special case of forest
algebras.  The states that are transformed, which in the case of
forest algebras correspond to the horizontal monoid, are the cartesian
product $H_1\times H_2$ with component-wise addition. The transforming
monoid, which in the case of forest algebra corresponds to the
vertical monoid, is more sophisticated, its carrier set is $ V_1\times
V_2^{H_1}$.  The action of the transforming monoid $ V_1\times
V_2^{H_1}$ on the transformed states $H_1 \times H_2$ is defined by
$$(v_1,f)(h_1,h_2)=(v_1h_1,f(h_1)h_2).$$
The composition  operation in the transforming monoid $V_1 \times V_2^{H_1}$ is defined by
\begin{align*}
(v,f) \cdot (v',f') = (vv', f'')  \qquad f''(h)=(f(v' h)) \cdot (f'(h)).
\end{align*}
As is well known, this definition turns $V_1\times V_2^{H_1}$ into a
monoid of faithful transformations on $H_1\times H_2.$ 
 (Observe that
since we define forest algebras using a left action of $V$ on $H,$
rather than a right action, our definition of the wreath product is
the reverse of the customary one, with the first algebra in the composition written as the left-hand factor in the wreath product, rather than as the right-hand factor.) 

By applying the definition of wreath product for transformation monoids, we have obtained most of the ingredients of forest algebra. We are missing the monoid operation on the horizontal monoid; for this we use the usual direct product.

The  last missing condition is that for every element $h$ of the horizontal monoid, a forest algebra should have elements $1+h$ and $h+1$ of the vertical monoid that satisfy 
\begin{align*}
	(1+h)g=g+h \qquad \mbox{and}\qquad (h+1)g=h+g.
\end{align*}We show that these elements exist in the wreath product. Let then $h=(h_1,h_2) \in H_1 \times H_2$. Consider the map $f:H_1\to V_2$ that sends every element to
$(1+h_2).$
 Then for any $g = (g_1,g_2) \in H_1 \times H_2,$ we have
\begin{eqnarray*} (1+h_1,f)(g_1,g_2)&=&((1+h_1)g_1, (1+h_2)g_2)\\
&=&(g_1+h_1,g_2+h_2)\\
&=& (g_1,g_2)+ (h_1,h_2).
\end{eqnarray*}
Therefore, the element $(1+h_1,f)$ plays the role of $1+(h_1,h_2)$.
Similarly, we find $V_1\times V_2^{H_1}$ contains the transformation
$(h_1,h_2)+1.$ 

Thus {\it the wreath product of two forest algebras is
  a forest algebra}. 

Well-known properties of the wreath product of transformation semigroups and monoids carry over unchanged to this setting.  In particular, the wreath product is associative, so we can talk about the wreath product of any sequence of forest algebras, and about the {\it iterated wreath product} of an arbitrary number of copies of a single forest algebra.  Likewise, the direct product of two forest algebras embeds in their wreath product in either direction. As a consequence, if $L_1,L_2$ are recognized by forest algebras $(H_1,V_1), (H_2,V_2)$ respectively, then their union and intersection are both recognized by $(H_1,V_1)\circ (H_2,V_2).$

The connection with sequential composition is given by:

\begin{thm}\label{thm:sequential-composition}
  For every pair of forest algebra homomorphisms
  \begin{eqnarray*}
\alpha : A^\Delta \to (G,W)\qquad      \beta : (A \times G)^\Delta \to (H,V)\ .
  \end{eqnarray*}
  there is a homomorphism into the wreath product $(G, W) \circ (H, V)$
  that, when restricted to forests, is equal to the sequential
  composition 
\begin{align*}
	\alpha \otimes \beta : H_A \to G \times H.
      \end{align*} Conversely, every homomorphism from a free forest
      algebra $A^{\Delta}$ into the wreath product of two forest
      algebras is realized in this manner by the sequential
      composition of two homomorphisms.
\end{thm}

\proof
Given homomorphisms $\alpha,$ $\beta$ as above, consider the map from $A$ into the vertical monoid of $(G,W)\circ (H,V)$ given by
$$a\mapsto (\alpha(a\hole),f_a),$$
where for all $a\in A,$ $g\in G,$
$$f_a(g)=\beta((a,g)\hole).$$
By the universal property of $A^{\Delta},$ this map extends to a unique homomorphism $\gamma$ with domain $A^{\Delta}.$  A straightforward induction on the construction of a forest $t\in H_A$ shows that $\gamma(t)=(\alpha(t),\beta(t^{\alpha}))$:  The crucial step is when $t=as$ for some $a\in A,$ $s\in H_A.$  We then have $t^{\alpha}=(a,\alpha(s))\cdot s^{\alpha},$ so that

\begin{eqnarray*}
\gamma(t) &=& \gamma(a)\cdot\gamma(s)\\
&=& (\alpha(a\hole),f_a)\cdot (\alpha(s),\beta(s^{\alpha}))\\
&=& (\alpha(a\hole)\cdot\alpha(s),\beta(f_a(\alpha(s)))\cdot\beta(s^{\alpha}))\\
&=& (\alpha(as),\beta((a,\alpha(s))\cdot \beta(s^{\alpha}))\\
&=& (\alpha(t),\beta(t^{\alpha})).
\end{eqnarray*}

Conversely, if $\gamma:A^{\Delta}\to(G,W)\circ(H,V)$ is a homomorphism, then for each $a\in A,$ $\gamma(a\hole))$ has the form $(w_a,f_a)$ for some $w_a\in W,$
$f_a:G\to V.$  We define homomorphisms 
\begin{eqnarray*}
\alpha : A^\Delta \to (G,W)\qquad      \beta : (A \times G)^\Delta \to (H,V)\ .
  \end{eqnarray*}
  by setting, for each $a\in A,$ $g\in G,$
  \begin{eqnarray*}
\alpha(a\hole)=w_a\qquad      \beta((a,g)\hole)=f_a(g)\ .
  \end{eqnarray*}
As we saw above, $\alpha\otimes\beta$ is the unique homomorphism mapping $a\hole$ to $(w_a,f_a),$ so $\gamma=\alpha\otimes\beta.$
\qed









\section{Wreath Product Characterizations of Language Classes}

When  $\algclass$ is a class of forest algebras, we write $\templ
\algclass$ for the class of languages recognized by iterated wreath
products of forest algebras from $\algclass$.  The following corollary
to Theorem~\ref{thm:sequential-composition} justifies this notation.

\begin{cor} Let $\langclass$ be the class of languages recognized by a
  class of forest algebras $\algclass$.  Then $\templ \langclass =
  \templ \algclass$.
\end{cor}

We also say that $\algclass$ is an \emph{algebraic base} of the
language class $\templ \algclass$ (note that there may be several
algebraic bases, just as there may be several language bases). We will
now exhibit algebraic bases for the logics discussed in
Section~\ref{section.logic}. By the above corollary, all we need to do
is to provide, for each logic, a class of forest algebras that
captures the language base. We could, of course, simply say that an algebraic base consists of the syntactic forest algebras of the members of the language base, but we prefer more explicit algebraic descriptions.  These are given in the following theorem;
 the algebras used in the statement are described immediately afterwards, while the detailed proofs are not given until Section 7. 
\begin{figure}
  \centering
  \begin{tabular}{l|l}
Logic & Algebraic base\\
\hline
EF & $     \mathcal{U}_1$\\
CTL & $    \mathcal{U}_2$\\
$FO[\prec]$ &     aperiodic path algebras\\
CTL*  &  distributive aperiodic algebras\\
 PDL &   distributive algebras\\
 graded PDL & path algebras
\end{tabular}
\caption{Algebraic bases for temporal logics}
  \label{fig:algebraic-bases}
\end{figure}

\begin{thm}\label{thm.main}
The logics EF, CTL, $FO[\prec]$, CTL*, PDL and graded PDL have algebraic bases as depicted
in Figure~\ref{fig:algebraic-bases}.
\end{thm}

We now proceed to describe the algebras mentioned in
Figure~\ref{fig:algebraic-bases}. The bases have been chosen so that
each base is either finite, or in the case it is an
infinite class of algebras, then it has an effective characterization,
i.e.~there is an algorithm that checks if the syntactic algebra
of a given forest language
belongs to the base. Furthermore, the infinite algebraic bases are
given by identities in the forest algebra, and therefore the algorithm
reduces to checking if the identities hold.

First, we recall that an {\it aperiodic} finite monoid $S$ is one
that contains no nontrivial groups. Equivalently, there exists $m>0$
such that $s^m=s^{m+1}$ for all $s\in S.$ When we say that a forest
algebra $(H,V)$ is aperiodic, we mean that the vertical monoid $V$ is
aperiodic (which implies that $H$ is aperiodic).

$\mathcal{U}_1$  is the forest algebra $(\{0,\infty\},\{1,0\}),$ with
$0\cdot \infty=0\cdot 0= \infty.$  Note that since we use additive
notation in the horizontal monoid, the additive absorbing element is
denoted $\infty,$ while the multiplicative absorbing element is 0. 
The vertical monoid of $\mathcal{U}_1$ is the unique smallest nontrivial
aperiodic monoid, denoted $U_1$ in the literature. 
Another description of  $\mathcal U_1$ is that it is  the syntactic forest algebra of the forest language ``some node with $a$'' over an alphabet $A \ni a$ with at least two letters. 
If follows that every language in
the language base of EF is recognized by
$\mathcal{U}_1,$ and every language  recognized by $\mathcal{U}_1$ is
a boolean combination of members of
the language base of EF, so this algebra forms an algebraic base for EF.

$\mathcal{U}_2$ is the forest algebra
$(\{0,\infty\},\{1,c_0,c_{\infty}\})$ with $c_h\cdot h'=h$ for all
horizontal elements $h,h'.$ If one reverses the action from left to
right and ignores the additive structure, $\mathcal{U}_2$ is the
aperiodic unit in the Krohn-Rhodes Theorem.  The underlying monoid of
this transformation semigroup is usually denoted $U_2.$ Every language
recognized by $\mathcal{U}_2$ is a boolean combination of members of the
language base of CTL, and all languages recognized by $\mathcal{U}_2$ are in CTL, so $\mathcal{U}_2$ forms an algebraic base for CTL.

So much for the singleton bases. We now describe the infinite bases.

A {\it distributive algebra} is a forest algebra $(H,V)$ such that $H$
is commutative and such that the action of $V$ on $H$
is distributive: $v(h_1+h_2)=vh_1+vh_2$ for all $v\in V,$ $h_1,h_2\in
H.$ The assertion that distributive algebras form algebraic bases for the given language classes is a consequence of the following theorem:

\begin{thm}\label{thm:distributive-algebra}
  A forest language is a boolean combination of languages $\tle L$
  (respectively, languages $\tle L$ with $L$ first-order definable) if
  and only if it is recognized by a distributive forest algebra
  (respectively, an aperiodic distributive forest algebra).
\end{thm}

Let us define a {\it path language}  to be any boolean combination of members of the language base
of graded PDL, and an {\it fo path language} to be a boolean combination of members of the language 
base of $FO[\prec].$  We have the following analogue to Theorem~\ref{thm:distributive-algebra}.

\begin{thm}\label{thm:path-algebra}
A finite forest algebra $(H,V)$ recognizes  only path languages if and only if $H$ is aperiodic and commutative
and 
\begin{equation}
  \label{eq:path-ident}
    vg + vh = v(g+h) + v0
\end{equation}

\begin{equation}
  \label{eq:fall-apart}
  u(g+h)  = u (g + uh)
\end{equation}
hold for all $g,h\in H$ and $u,v\in V$ with $u^2=u.$  $(H,V)$ recognizes only
fo-path languages if and only if $H$ is aperiodic and commutative, $V$ is aperiodic,
and $(H,V)$ satisfies the two identities above.
\end{thm}

We define a {\it path algebra} to be a forest algebra $(H, V)$
satisfying identities \ref{eq:path-ident} and \ref{eq:fall-apart} with
$H$ aperiodic and commutative.  We will give the proofs of Theorems~\ref{thm:distributive-algebra} and \ref{thm:path-algebra} in Section~\ref{sec:path-algebras-path}. 

Because of the connection with logic, we will call divisors of the
six kinds of iterated wreath products described above EF-algebras,
CTL-algebras, CTL*-algebras, FO-algebras,  PDL-algebras, and graded PDL-algebras,
respectively.

Note that for EF and CTL, the algebraic base has one algebra, while
our other bases contain infinitely many algebras. This turns out to be
optimal, as stated below. 
\begin{thm}[Infinite base theorem]\label{thm.nofinitebase}
  None of the language classes CTL*, $FO[\prec],$ PDL, or graded PDL has a finite algebraic base.
\end{thm}

\proof
If a language class has an algebraic base consisting of a finite set of forest algebras
$$(H_1,V_1),\ldots,(H_k,V_k),$$
then it has a base containing just the single algebra
$$(H,V)=(H_1,V_1)\times\cdots\times (H_k,V_k).$$
This is because each of the $(H_i,V_i)$ divides $(H,V),$ and $(H,V)$ embeds into the wreath product of the $(H_i,V_i),$ in any order.  Consequently, iterated wreath products of the $(H_i,V_i)$ and iterated wreath products of $(H,V)$ have the same divisors, and so recognize the same languages.  

By these observations,  it suffices to show that none of the classes in the statement of the theorem has an algebraic base consisting of a single forest algebra $(H,V).$  We will give two different arguments for this, one applicable to the aperiodic classes CTL* and $FO[\prec],$ and the other for the nonaperiodic classes.

Suppose the language class $FO[\prec]$ is generated by a single
  algebra $(H,V).$ Since $(H,V)$ is required to recognize only
  languages in this class, $V$ is aperiodic, and thus there is an
  integer $n$ such that $v^n =v^{n+1}$ for all $v\in V.$ We will show
  that no iterated wreath product of copies of $(H,V)$ can recognize
  the language $L_n$ consisting of all forests over $A=\{a,b,c\}$ in
  which there is a path from the root with the label in $(a^nb)^*c.$
  Since $L_n$ is in $CTL^*\subseteq FO[\prec],$ this will give the
  desired conclusion also for $CTL^*$.

  We prove this by induction on the number of factors $k$ in the
  wreath product, showing that there are forests $s_k\in L_n$ and
  $t_k\notin L_n,$ such that $\phi(s_k)=\phi(t_k)$ is satisfied for
  every homomorphism $\phi$ from $A^{\Delta}$ into the $k$-fold wreath
  product of $(H,V)$.  For $k=1,$ we can simply take $s_1=a^nbc$ and
  $t_1=a^{n+1}bc.$ For the inductive step we suppose the claim holds
  for some $k\ge 1,$ and let $(G,W)$ denote the $k$-fold wreath
  product of the $(H,V).$ Consider a homomorphism $\phi$ from
  $A^{\Delta}$ into the $(k+1)$-fold wreath product $(G,W)\circ
  (H,V).$ Recalling the definition of the wreath product we have $\phi:
  A^\Delta \to (G\times H, W\times V^G)$. If we compose $\phi$ with
  the projection onto the left coordinate we obtain a homomorphism
  $\psi$ into $(G,W).$ Note that since aperiodicity is preserved under
  wreath products, there is an $m$ such that $w^m=w^{m+1}$ for all
  $w\in W.$

We first claim that if $p$ and $q$ are contexts in $V_A$ such that $\psi(p)=\psi(q),$ then
$$\phi(p^nq^{m+1})=\phi(p^{n+1}q^{m+1}).$$
To see this, first take $(g_0,h_0)$ in $G\times H.$  We have 
\begin{align*}
	\psi(q^{m})g_0=\psi(q^{m+1})g_0=\psi(p)\psi(q^m)g_0,
\end{align*}
so we have
$$\phi(q^{m+1})(g_0,h_0)=(g_1,h_1),$$ where
$\psi(p)g_1=g_1.$ Let us write $\phi(p)$ as $(\psi(p),f),$ where
$f:G\to V.$ We then have
$$\phi(p^n)(g_1,h_1)= (g_1,f(g_1)^nh_1)=(g_1,f(g_1)^{n+1}h_1)=\phi(p^{n+1})(g_1,h_1).$$ Since $g_0, h_0$ are arbitrary, this proves $\phi(p^nq^{m+1})=\phi(p^{n+1}q^{m+1}),$ as claimed. We now make particular choices for $p$ and $q,$ namely

$$p=a\hole + bt_k,\quad q=a\hole+bs_k.$$
Since $\psi(s_k)=\psi(t_k),$ we have $\psi(p)=\psi(q),$ and thus by
our claim above, $\phi(p^nq^{m+1})=\phi(p^{n+1}q^{m+1}).$ Set
$s_{k+1}=p^nq^{m+1}\cdot 0,$ and
  $t_{k+1}=p^{n+1}q^{m+1}\cdot 0.$ So
$\phi(s_{k+1})=\phi(t_{k+1}).$ For every path $w$ from the root in
$s_k,$ there is a path in $s_{k+1}$ with label $a^nbw.$ 
On the other
hand, for every path with a label $a^nbv$ from the root of $t_{k+1}$ we
have $v\in t_k$. Thus
$s_{k+1}\in L_n$ and $t_{k+1}\notin L_n,$ as claimed.

We now turn to the nonaperiodic case. Let $p$ be a prime that does not divide the order of any group in $(H,V),$ and let $L$ be the set of forests over $\{a,b\}$ in which there is a path from the root of the form $a^mb,$ where $p$ divides $m.$  We will show that $(H,V)$ cannot recognize $L.$  Since $L$ has the form $EK$ for a regular word language $K,$ $L$ is in PDL, so this will complete the proof.

It is easy to see that the vertical monoid of the syntactic forest algebra of $L$ contains a group of order $p$:  Let $0\le r <p,$ and let $H_r$ be the set of forests in which every path from the root has an initial segment of the form $a^jb,$  where $r=j\bmod p.$ Each $H_r$ is a class of the syntactic congruence, all $p$ of these classes are distinct, and the context $a\hole$ cyclically permutes them. On the other hand, the set of simple groups dividing a transformation monoid is preserved under wreath product, so no iterated wreath product of copies of $(H,V)$ can contain a group of order $p,$ and thus cannot recognize $L.$

\qed




\section{{\bf EF}}\label{section.ef}

The logic EF was one of the first logics over trees to have a
decidable characterization~\cite{bw}. The result has been since then
reproved several times with different
methods~\cite{Wu-IPL07,ES07}. Here we give a new proof based on wreath
product. Our argument is purely algebraic. It computes a decomposition
based on the ideal structure of the underlying forest algebra.

The following theorem is proved in ~\cite{bw}.
\begin{thm}\label{thm.efcharacterization}  A forest language  $L\subseteq H_A$ is defined by a forest formula of  EF if and only if (i) $H_L$ is idempotent and commutative, and (ii) for every $v\in V_L,$ $ h\in H_L,$ we have $vh+h=vh.$
\end{thm}

Because this property can be effectively verified from the
multiplication tables of $H_L$ and $V_L,$ we have an effective characterization of EF. More specifically,  there is a decision procedure
for determining whether or not a forest language given, say, by an
automaton that recognizes it, is definable by a forest formula of EF.   This procedure can
also be adapted to testing whether a tree language is EF-definable
with tree semantics.

In light of Theorem~\ref{thm.main}, Theorem~\ref{thm.efcharacterization} can be formulated as follows. 

\begin{thm}\label{thm.efdecomposition} A forest algebra $(H,V)$ divides an iterated wreath product of copies of $\mathcal{U}_1$ if and only if $H$ is idempotent and commutative, and $vh+h=vh$ for all $h\in H,$ $v\in V.$
\end{thm}

Note that Theorem~\ref{thm.efdecomposition} is   purely algebraic. It makes no mention of trees, forests, languages or logic.  This suggests that it might be proved reasoning solely from the structure of the forest algebra. 

Here we present such a proof. The easy direction is to show that every
divisor of an iterated wreath product of copies of $\mathcal{U}_1$ is
horizontally idempotent and commutative and satisfies the identity $vh+h=vh$.  Identities are always preserved under division, and obviously $\mathcal{U}_1$ itself satisfies the properties, so we just need to show that the properties are preserved under wreath product.  Let $(G,W)$ and $(H,V)$ be forest algebras satisfying the identity, with $G,H$ idempotent and commutative. The horizontal monoid of the wreath product is just $G\times H,$ which is idempotent and commutative. Let $h=(h_0,h_1)\in (G,W),$ $v=(v_0,f)\in W\times V^G$ be horizontal and vertical elements of the wreath product.  We have
\begin{eqnarray*}vh+h &=& (v_0,f)(h_0,h_1)+(h_0,h_1)\\&=&(v_0h_0+h_0,f(h_0)h_1+h_1)\\&=&(v_0h_0,f(h_0)h_1)\\&=&
(v_0,f)(h_0,h_1)\\&=&vh.
\end{eqnarray*}

For the converse, we suppose $(H,V)$ is horizontally idempotent and commutative and satisfies the identity. We prove by induction on $|H|$ that $(H,V)$ divides an iterated wreath product of copies of $\mathcal{U}_1.$ 
%

Since $H$ is idempotent and commutative, it is partially ordered by the relation $\leq$  defined by $h_1\leq h_2$ if and only if $h_1=h_2+h$ for some $h\in H.$  Transitivity and reflexivity of this relation are obvious.  Antisymmetry follows from the observation that if $h_1=h_2+h\leq h_2,$ then $h_1+h_2=h_2+h+h_2=h_2+h=h_1.$ Thus if we have both $h_1\leq h_2$ and $h_2\leq h_1,$ then $h_1=h_1+h_2=h_2.$  This is just the standard $\mathcal{J}$-ordering, one of the Green relations, on the monoid $H.$  Thus our identity $vh+h=vh$ implies $vh\leq h$ for all $v\in V,$ $h\in H.$  Conversely, if $vh\leq h,$ then there is some $h'\in H$ such that $vh=h+h',$ and thus $vh+h=h+h'+h=h+h'=vh.$  So we can replace the identity by the inequality $vh\leq h$ for all $v\in V,$ $h\in H.$

The sum of all the elements of $H$ is the (necessarily unique)
absorbing element, which, following our usual practice, we denote
$\infty.$ This is the unique $\leq$-minimal element, since obviously
$\infty + h =\infty$ for all $h\in H.$ If $|H|\leq 2,$ then $(H,V)$ is
either trivial, or isomorphic to $\mathcal{U}_1,$ so we can assume
$|H|>2.$ Thus there is at least one minimal element $h\not=0$ in
$H\setminus\set{\infty}$. We call such an element a {\it subminimal}
element. It has the property that for all
$v\in V,$ $vh=h$ or $vh=\infty.$

For
each subminimal $h,$ we define $H_h$ to be the {\it set}
$\{\infty\}\cup\{g:h\in Vg\}.$ Observe that  $H_h$ is a submonoid of $H,$ because if $v_1h_1=h$ and $v_2h_2=h,$ then 
\begin{align*}
	h=h+h=v_1h_1+v_2h_2+h_1+h_2=u(h_1+h_2) \qquad \mbox{where }u=v_1h_1+v_2h_2+1.
\end{align*}  For $v\in V$ and $g\in H_h$ we set $v*g=vg$ if $vg\in H_h,$ and otherwise set $v*g=\infty.$  It is straightforward to verify that 
for all $v_1,v_2\in V,$ $g\in H_h,$
$$v_1*g+g=v_1*g,$$
$$(v_1v_2)*g=v_1*(v_2*g),$$
so we get a well-defined action of $V$ on $H_h.$  We can collapse this action to make this faithful, and thus we get a well-defined forest algebra $(H_h,V_h),$ that satisfies the hypotheses of the theorem.  If there is more than one subminimal element, then each $H_h$ has strictly smaller cardinality than $H.$  Further, consider the map
$$\iota:(H,V)\to\prod (H_h,V_h),$$
where the direct product is over all subminimal elements $h,$ defined by setting the $h$-component of $\iota(g)$ to be $g$ if $g\in H_h,$ and $\infty$ otherwise.  It is straightforward to verify that $\iota$ is a homomorphism embedding $(H,V)$ into the direct product.  Since the direct product in turn embeds into the wreath product, we get the result by the inductive hypothesis.

It remains to consider the case where there is just one subminimal element $h.$  In this case $(H_h,V_h)$ is identical to $(H,V).$  The elements of $H$ different from $\infty$ form a submonoid $G$ of $H.$  We get a well-defined action ${**}$ of $V$ on $G$ by setting $v{**}g=vg$ if $g\in G,$ and $v{**}g=h$ otherwise.  Once again, the resulting forest algebra $(G,W)$ satisfies the necessary identities, so by the inductive hypothesis $(G,W)$ divides a wreath product of copies of $\mathcal{U}_1.$  We complete the proof by showing that $(H,V)$ embeds in the wreath product
$(G,W)\circ\mathcal{U}_1.$ We map $g\in H-\{\infty\}$ to $\alpha(g)=(g,0)$ and $\infty$ to
$\alpha(\infty)=(h,\infty).$ We further map $v\in V$ to $(v,f_v),$
where $f_v(g)=0$ if $vg=\infty,$ and $f_v(g)=1,$ otherwise. This is
obviously an injective homomorphism on the additive structure.  
To show that it is a homomorphism on the multiplicative structure, it suffices to show that for all $v\in V,$ $g\in H,$ $\alpha(vg)=(v,f_v)\alpha(g).$  There are several cases to consider.  First, if $g=\infty,$ then $vg=\infty,$ so we have
$$\alpha(vg)=(h,\infty)=(v{**}h,f_v(h)\infty)=(v,f_v)(h,\infty)=(v,f_v)\alpha(g).$$
If $g\neq\infty$ but $vg=\infty,$ we have
$$\alpha(vg)=(h,\infty)=(v{**}g,0\cdot 0)=(v{**}g,f_v(g)\cdot 0)=(v,f_v)(g,0)=(v,f_v)\alpha(g).$$
Finally, if neither $g$ nor $vg$ is $\infty,$ we have
$$\alpha(vg)=(vg,0)=(v{**}g,1\cdot 0)=(v,f_v)(g,0)=(v,f_v)\alpha(g).$$
\qed

Theorem~\ref{thm.efdecomposition} is the exact analogue for forest algebras of a Theorem of Stiffler~\cite{stiffler} showing that a finite monoid is $\mathcal{R}$-trivial if and only if it divides a wreath product of copies of $\mathcal{U}_1.$  Because of our conventions on the direction of the action, all our EF-algebras have $\mathcal{L}$-trivial, rather than $\mathcal{R}$-trivial vertical monoids.



\section{Path Algebras and Distributive Algebras}
\label{sec:path-algebras-path}
In this section we prove
Theorems~\ref{thm:distributive-algebra} and~\ref{thm:path-algebra}. 

\subsection{Distributive algebras}
We begin with Theorem~\ref{thm:distributive-algebra}, whose proof is
significantly simpler than the proof of
Theorem~\ref{thm:path-algebra}.  Recall that a {\it distributive
  algebra} is a forest algebra $(H,V)$ where $H$ is commutative and
which satisfies
\begin{eqnarray*}
  v(h_1+h_2)=vh_1+vh_2\ .
\end{eqnarray*}
Note that instead of the two requirements, horizontal commutativity
and the above identity, we could use a single identity
\begin{eqnarray*}
  v(h_1+h_2)=vh_2+vh_1\ ,
\end{eqnarray*}
which, when $v=1$, gives also horizontal
commutativity. Nevertheless, we prefer separating the two conditions.

Theorem~\ref{thm:distributive-algebra} says that a forest language is
a boolean combination of languages $\tle L$ (respectively, languages
$\tle L$ with $L$ first-order definable) if and only if it is
recognized by a distributive forest algebra (respectively, an
aperiodic distributive forest algebra). 

The ``only if''part is fairly straightforward, applying any of the
identities required from a distributive algebra does not change the
set of paths in a tree. For  the ``if'' part, only a little bit of
effort is needed. The idea is that by applying the conditions on
distributivity, one can show that if $\alpha$ is a homomorphism into a
distributive algebra, then a forest is equal to the sum of its paths. More precisely, if $t$ is a forest with nodes $x_1,\ldots,x_n$ then 
\begin{eqnarray}\label{eq:t-sum}
  \alpha(t)=\alpha(t_1 + \cdots + t_n)
\end{eqnarray}
where each tree $t_i$ is obtained by taking the node $x_i$  and
removing all nodes from $t$ that are not ancestors of $x_i$. This is depicted in the
picture below.
\begin{center}
  \includegraphics[scale=0.7]{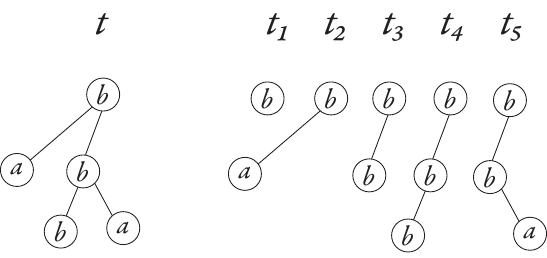}
\end{center}
Note that $H$, apart from being a commutative monoid, is also
idempotent, by
\begin{eqnarray*}
h=  (h+1)(0 + 0) = (h+1)0 + (h+1)0= h+h\ .
\end{eqnarray*}
In particular, the  value  of 
\begin{align*}
\alpha(t) = 	\alpha(t_1 + \cdots + t_n) = \alpha(t_1)+ \cdots + \alpha(t_n)
\end{align*}
does not depend on the order or multiplicity of types in the sequence $\alpha(t_1),\ldots,\alpha(t_n)$, and only on the set of values  $\set{\alpha(t_1),\ldots,\alpha(t_n)}$. For each $g \in H$, we define the word language
\begin{eqnarray*}
  L_g= \set{a_1 \cdots a_i \in A^* : \alpha(a_1 \cdots a_i 0)=g}.
\end{eqnarray*}
It is not difficult to see that a forest $t$ from~\eqref{eq:t-sum} satisfies the formula $\tle L_h$ if and only if one of the types $\alpha(t_1),\ldots,\alpha(t_n)$ is $g$.
Combining the observations above, we conclude  that for every $h \in H$
\begin{align*}
	\alpha(t)={g} \qquad \mbox{iff} \qquad  \bigvee_{\substack{G \subseteq H \\ h =  \sum_{g \in G}g}} \big( \bigwedge_{g \in G} \tle L_g \quad \land \quad  \bigwedge_{g \in H-G} \neg \tle L_g \big).
\end{align*}
Furthermore, if the monoid $V$ is aperiodic, each word  language $L_h$, as a word language recognized by  $V$, is
first-order definable by the McNaughton-Papert theorem.

\subsection{Path algebras}
We now proceed to prove Theorem~\ref{thm:path-algebra}.  We use the
term \emph{path algebra} for a forest algebra that satisfies the
conditions in the theorem, namely that the horizontal monoid is  aperiodic and
commutative, along with identities identities~\eqref{eq:path-ident}
and~\eqref{eq:fall-apart}, which we recall here
\begin{equation*}
	\tag{\ref{eq:path-ident}}  
    vg + vh = v(g+h) + v0 
\end{equation*}

\begin{equation*}
  \tag{\ref{eq:fall-apart}}
  u(g+h)  = u (g + uh) \qquad\text{for $u$ s.t. $u^2=u$}
\end{equation*}
Recall that a path language is a boolean combination of the
  languages 
from the base of graded PDL: languages of the form ``at least $k$
paths in $L$, for some regular $L$. A first-order definable path language is defined
similarly but $L$ is required to be definable in $FO_A[<]$.

 Theorem~\ref{thm:path-algebra} says that a
forest language is a path language (respectively, a first-order definable path language) if
and only if it is recognized by a path algebra (respectively, a
vertically aperiodic path algebra).

The ``only if'' part is simple; the identities are designed to hold in
any syntactic algebra of a path language (respectively, a first-order
definable path language). The rest of
this section is devoted to showing the ``if''
implication of the theorem.

For the moment, we concentrate on path algebras, as opposed to
aperiodic path algebras. After doing the proof, we show how it can be
modified to obtain the case for aperiodic path algebras.

We begin with the following lemma, which illustrates the significance
of identity~\eqref{eq:path-ident}. When speaking of paths, we refer to paths that begin in one of the roots of a forest, but that end in any node, not necessarily a leaf.

\begin{lem}\label{lemma:multiset}
  Forests with the same multisets of paths have the same image under
  any homomorphism into an algebra satisfying~\eqref{eq:path-ident} and
  horizontal commutativity.
\end{lem}
\proof
  We will show that two forests with the same multisets of paths are
  equal in the quotient of the free forest algebra under the
  identities~\eqref{eq:path-ident} and $h+g=g+h$. In other words, we
  show that if~\eqref{eq:path-ident} and $h+g=g+h$ are treated as
  rewriting rules on real forests (and not elements of the
    forest algebra), then each
  two forests with the same multisets of paths can be rewritten into
  each other.  The idea is to transform each forest into a normal
  form, such that the normal form is uniquely determined by the
  multiset of paths. The transformation into normal form works as
  follows. Let $t$ be a forest. Let $a_1,\ldots,a_n$ be the labels
  that appear in the roots of $t$. By applying horizontal
  commutativity, the forest $t$ is rewritten into a forest
  \begin{eqnarray*}
\sum_i a_it_{i,1} + a_it_{i,2} + \cdots + a_it_{i,n_i}.
  \end{eqnarray*}
  By applying the identity~\eqref{eq:path-ident} and horizontal commutativity, the above is
  rewritten into
  \begin{eqnarray*}
\sum_i a_i(t_{i,1} + t_{i,2} + \cdots + t_{i,n_i}) +  \overbrace{a_i0+ \cdots +a_i0}^{(n_i-1) \mbox{ times}}
  \end{eqnarray*}
  Finally, for each $i$, we rewrite the forest $t_{i,1} + t_{i,2} +
  \cdots + t_{i,n_i}$ into normal form. The result of this rewriting
  is a forest where every two different non-leaf nodes have a
  different sequence of labels on their paths. Such a forest, modulo
  commutativity, is uniquely determined by the multiset of paths.
\qed

The above lemma shows membership in a language $L$ recognized by an
algebra satisfying~\eqref{eq:path-ident} is uniquely determined by the
multiset of paths in a forest.  However, this on its own does not mean
that $L$ is a path language (otherwise, we would not need
identity~\eqref{eq:fall-apart}), as witnessed by the following
example.

\bigskip
\noindent{\bf Example. }
Consider the language $a^*(a+a)$. A forest belongs to this language if
and only if for some $n \in \Nat$, the multiset of paths is
\begin{eqnarray*}
  \epsilon, a, a^2, \ldots, a^n, a^{n+1}, a^{n+1}
\end{eqnarray*}
This language is not a path language. It does not even belong to a quite general class defined
below. Let $\alpha : A^* \to M$ be a morphism from words into a finite
monoid. The $\alpha$-profile of a forest $s$ is a vector in $\Nat^M$
that says, for each $m \in M$, how many times a path with value $m$
appears in the forest $s$. A language is called \emph{path-profile
  testable} if for some morphism $\alpha$, membership in the language
is uniquely determined by the $\alpha$-profile of a forest. It is not
difficult to see that the language $a^*(a+a)$ is not even path-profile
testable, since a path-profile testable language will confuse $a^n(a+
a^na)$ with $a^n a^n(a+a)$ for certain large values of $n$ (more
precisely for $n=\omega$, the notion of $\omega$ will be defined
below).
\bigskip

We now return to proving the ``if'' implication in
Theorem~\ref{thm:path-algebra}. The theorem follows immediately from the
proposition below, by taking $v$ to be the empty context.
\begin{prop}\label{prop:ind-prop}
  Let $(H,V)$ be a path algebra. For any $v \in V$ and $h \in
  H$, the forest language $\set{t : v\alpha(t)=h}$ is a path language.
\end{prop}
For the rest of this section we fix a path algebra $(H,V)$
  and a homomorphism $\a: A^\Delta\to (H,V)$. For a tree $t$ 
  we will often refer to $\a(t)$ as \emph{type of $t$}. Similarly for
  contexts. 

We will prove the proposition by induction on the size of the set $vV
\subseteq V$. We write $v \sim w$ if $vV=wV$ (this is Green's
$\Rr$-equivalence in the context monoid).

Apart from Green's relations, we will also use the $\omega$ power from
monoid theory. For a finite monoid  -- in this case, the monoid is $V$ --
we define $\omega$ to be a number such that $v^\omega$ is idempotent
for any $v \in V$. Such a number always exists in a finite monoid, it
suffices to take $\omega$ to be the factorial of the size of $V$.

The induction base is when the set  $vV $ is  minimal. 
\begin{lem}
	If $vV$ is minimal, then the context $v$ is constant, which means that $vg=vh$ holds for every $g,h \in H$
\end{lem}
\proof

Let $h_1,\ldots,h_n$ be all elements of $H$. Consider the context 
\begin{align*}
	w = v(h_1+\ldots+h_n+1).
\end{align*}
We show that the context  $w$ is constant.   It suffices to show that $wh = w0$ for every $h \in H$. Because $h_1,\ldots,h_n$ contains all elements of $H$, then there must be some $i$ such that $h_i$ is $\omega \cdot h$, which is defined by
\begin{align*}
	 \omega \cdot h = \overbrace{h+ \cdots + h}^{\text{$\omega$ times}}.
\end{align*}
By aperiodicity of $H$, we know that $h_i+h = h_i$. By commutativity of $H$, we see that 
\begin{align*}
	h_1+ \cdots + h_n = h_1 + \cdots + h_n + h
\end{align*}
and therefore $wh=w0$. We have thus established that $w$ is constant. Since $w$ is constant, one can easily see that $wu=w$ holds for every $u \in V$, and therefore $wV = \set w$. Since $w \in vV$, it follows that $wV \subseteq vV$. By minimality of $vV$, we infer that $vV = \set w$. Because $V$ contains an identity context, it follows that $v \in vV$ and therefore $w=v$, and therefore $v$ is constant.
\qed
 When the context
$v$ is constant, the language in  the proposition is either empty, or
all forests, in either case it is a path language.

We now proceed to the induction step. We fix $v$ and $h$ as in the
statement of the proposition.

A \emph{path context} is a context of the form $a_1\cdots a_n\hole$.  A
\emph{preserving context} is a context $p$ whose type satisfies
$v\alpha(p)=v$, for our fixed $v$.  A forest is called \emph{negligible} if it is a
concatenation of trees of the form $p0$, where $p$ is a preserving context. 

\begin{lem}\label{lemma:stab}
  If $vu \sim v$ then $v=v(1 + u0)$. In particular, if $g$ is the
  type of a negligible forest, then $v=v(g+1)$.
\end{lem}
\proof
 In the proof, we will use the identity
\begin{equation}
  \label{eq:add-side}
  v^\omega=v^\omega(1 + v^\omega0)\ .
\end{equation}
Note that by  iterating the above $\omega$ times, we get
\begin{equation}
  \label{eq:add-side-omega}
  v^\omega=v^\omega(1 + \omega \cdot v^\omega0)\ .
\end{equation}
In the above, $\omega \cdot v^\omega 0$, or more generally $n \cdot h$ for any number $n$ and forest type $h$, denotes the $n$-fold sum $h + \cdots + h$.
The proof of~\eqref{eq:add-side} is by applying the
identity~\eqref{eq:fall-apart} from the definition of path algebras:
  \begin{eqnarray*}
    v^\omega g=v^\omega(g+0)=v^\omega(g+v^\omega 0)\ .
  \end{eqnarray*}

  We now proceed to prove the lemma.  If $vu \sim v$ then $vuw=v$ for
  some $w$. We can assume that $uw$ is idempotent, by replacing $uw$
  with $(uw)^\omega$.  By the identity~\eqref{eq:add-side-omega}, we
  get
\begin{eqnarray*}
  v = vuw=v(uw)(uw+\omega uw0) = v(uw + \omega uw0)\ .
\end{eqnarray*}
Let $x = (uw + \omega uw0)$. We know that $x=x+ \omega uw0$, and
therefore also $x^\omega = x^\omega + \omega uw0$. 
\begin{eqnarray*}
  v = vx^\omega = vx^\omega (   x^\omega + \omega uw0) \ .
\end{eqnarray*}
By applying~\eqref{eq:fall-apart} the above becomes
\begin{eqnarray*}
vx^\omega(1
  + \omega uw0) = v( 1 + \omega uw0).
\end{eqnarray*}
If we can show that $\omega uw0 = \omega uw0 + u0$, then we would be
done, by
\begin{eqnarray*}
  v = v(1 + \omega uw0) = v(1 + \omega uw0 + u0) = v(1 +
  u0).
\end{eqnarray*}
It remains to show  $\omega uw0 = \omega uw0 + u0$:
\begin{eqnarray*}
  \omega uw0 = \omega uw0 + \omega uw0 \stackrel{\eqref{eq:path-ident}}= \omega u(w0 + w0) + \omega u0
\end{eqnarray*}
and the last expression is clearly invariant under adding $u0$.
\qed

A \emph{guarded context} is a path context $pa\hole$ where the prefix
$p\hole$ is preserving, but the whole context $pa\hole$ is not.  A
forest is in \emph{guarded form} if it is a concatenation of trees of
the form $pat$, where $pa\hole$ is a guarded context.  The following
lemma shows that, up to negligible forests, each forest has the same
multiset of types as some guarded context.

We say two forests $t,t'$ are \emph{negligibly equivalent} if for some
negligible forests $s$ and $s'$, the forests $t+s$ and
$\impadded{t'}+s'$ have the same multiset of paths. This is indeed
an equivalence relation (it is transitive since a concatenation of
negligible forests is also negligible).
\begin{lem}\label{lemma:negligible}
  Each forest is negligibly equivalent to a guarded forest.
\end{lem}
\proof Let $t$ be a forest.  For each node $x$ in $t$, let $q_x$ be
the path context $a_1 \cdots a_m \hole$ obtained by reading the path
that leads to $x$ inside $t$, including the node $x$ (which has the
last label $a_m$). Let $X$ be the set of nodes $x$ for which the
context $q_x$ is a guarded context, in particular, $q_x =p_x a_x
\hole$, with $p_x$ a preserving path context and $a_x \in A$. Note
that the set $X$ is an antichain: a node $x \in X$ is
chosen as the first time when the path leading to $x$ stops preserving
$v$.  For each $x \in X$, let $t|_x$ be the subtree of the node $x$,
the subtree includes $x$. Let $t'$ be the
forest obtained from $t$ by removing all subtrees $t|_x$, for $x \in
X$. The forests
  \begin{eqnarray*}
    t + \sum_{x \in X} p_x0  \qquad \mbox{and} \qquad t'  + \sum_{x \in X} p_x t|_x
  \end{eqnarray*}
  clearly have the same multisets of paths. Since $\sum_{x \in X}
  p_x0$ is a negligible forest, and $\sum_{x \in X} p_x t|_x$ is a
  guarded forest, it remains to prove that $t'$ is negligibly
  equivalent to the empty forest. But this follows since all paths
  inside $t'$ correspond to preserving contexts, by construction of
  $t'$
\qed

A path language $L$ is called  \emph{guarded} if it is invariant
under concatenation with negligible forests, i.e.
\begin{eqnarray*}
  t + s  \in L  \qquad \mbox{iff} \qquad t  \in L
\end{eqnarray*}
holds for any negligible forest $s$. 

\begin{lem}\label{lemma:guarded-lemma}
  For any $h \in H$ there is a guarded path language $L_h$ such that
  for any guarded forest $t$,
  \begin{equation}
    \label{eq:guarded-formula}
        t \in L_h \qquad\mbox{iff}\qquad v\alpha(t)=vh\ .
  \end{equation}
\end{lem}

Before showing the lemma above, we show in the lemma below that  it concludes the proof of
Proposition~\ref{prop:ind-prop}. 

\begin{lem}
  Let $h$ and $L_h$ be as in Lemma~\ref{lemma:guarded-lemma}. Then the
  equivalence~(\ref{eq:guarded-formula}) holds for all forests $t$,
  and not only guarded forests.
\end{lem}
\proof
  Let $t$ be a forest.  By applying Lemma~\ref{lemma:negligible}, we
  can find negligible forests $s,s'$ and a guarded forest $t'$ such
  that $t+s$ and $t'+s'$ have the same multiset of paths.

  We begin with the left to right implication
  in~(\ref{eq:guarded-formula}). Assume that $t \in L_h$. Since
  $s$ is negligible and the language $L_h$ is guarded, $L_h$ also
  contains $t+s$. Since $L_h$ is a path language, and the forests $t+s$ and $t'+s'$ have the same multiset of paths, then $L_h$  also contains
  $t'+s'$. We now apply Lemma~\ref{lemma:guarded-lemma} to conclude
  that $v\alpha(t' + s')=vh$. Since $t' + s'$ and $t+s$ have the same
  multiset of paths, they have the same value under $\alpha$ by
  Lemma~\ref{lemma:multiset}. This gives us
  \begin{eqnarray*}
    vh =v\alpha(t'+s')=v\alpha(t+s)=v\alpha(t)\ ,
  \end{eqnarray*}
  where the last equality is by Lemma~\ref{lemma:stab}.

  The right to left implication is by reversing the above reasoning.
\qed

\subsection{The path language $L_h$}
We are only left with proving Lemma~\ref{lemma:guarded-lemma}.

We say that two types $g,h$ are $v+$-equivalent if $vug=vuh$ holds for any context $vu \not \sim v$. 
\begin{lem}\label{lemma:induction-assumption}
  For every $h \in H$, there is a path language $M_h$ that contains
  all forests whose type is $v+$-equivalent to $h$.
\end{lem}
\proof
Define 
\begin{align*}
	W = \set { w \in vV : w \not \sim v}
\end{align*}
By definition, a forest type $g \in H$ is $v+$-equivalent to $h$ if and only if $wg=wh$ holds for all $w \in W$. By the induction assumption in Proposition~\ref{prop:ind-prop}, we know that for every $w \in W$ and $g \in H$, the forest language
\begin{align*}
	L_{w,g} = \set{t : w \cdot \alpha(t) = g}
\end{align*}
is a path language.  The type of a forest $t$ is $v+$-equivalent to $h$ if for every context $w \in W$, the result of placing $t$ in a context of type $w$ is the same as the result of placing $h$ in $w$. In other words, the set of forests whose type is $v+$-equivalent to $h$ is 
\begin{align*}
	\bigcap_{w \in W} L_{w,wh},
\end{align*}
which is a path language, as an intersection of path languages.
\qed

Below, we write $guard$ for the set of pairs $(w,a) \in V \times A$
such that $vw \sim v$  but $vw\alpha(a) \not \sim v$. In other
words, a pair $(w,a) \in guard$ describes a guarded context $pa$.
Consider a forest in guarded form
\begin{eqnarray*}
t=       p_1a_1t_1  + \cdots + p_na_nt_n\ .
\end{eqnarray*}
For each $(w,a) \in guard$, let $I_{w,a}$ be the set of indexes $i$
such that $\alpha(p_i)=w$ and $a_i=a$.  The \emph{guarded profile} of
this forest is the function
\begin{eqnarray*}
  \tau_t \quad : \quad guard \quad \to \quad \set{0,1,\ldots,\omega} \times [H]_{\equiv_{v+}}
\end{eqnarray*}
which maps a pair $(w,a)$ to the pair $(n,h)$, where $n$ is the size
of $I_{w,a}$ (up to threshold $\omega$) and $h$ is the equivalence
class 
\begin{eqnarray*}
 [ \alpha(\sum_{i \in I_{w,a}} t_i)]_{\equiv_{v+}}\ .
\end{eqnarray*}

Lemma~\ref{lemma:guarded-lemma}, and thus also
Proposition~\ref{prop:ind-prop} and Theorem~\ref{thm:path-algebra},  will
follow from the two lemmas below.

\begin{lem}\label{lemma:profile-is-enough}
  For a guarded forest $t$, the guarded profile determines the value
  $v\alpha(t)$. In other words, if $s,t$ are guarded forests with the
  same guarded profile, then $v\alpha(s)=v\alpha(t)$.
\end{lem}

\begin{lem}\label{lemma:profile-can-be-described}
  For a guarded forest, the guarded profile can be determined by a
  path language. In other words, for each guarded profile $\tau$,
  there is a path language $L_\tau$ such that $t \in L_\tau \iff
  \tau_t = \tau$ holds for all guarded forests $t$.
\end{lem}

The above two lemmas give us Lemma~\ref{lemma:guarded-lemma}, by
taking $L_h$ to be the union of all $L_\tau$, for profiles $\tau$ of
the form $\tau=\tau_t$ where $t$ is a guarded forest with
$v\alpha(t)=h$.  We begin with the proof of
Lemma~\ref{lemma:profile-is-enough}.

\proof {\it (of Lemma~\ref{lemma:profile-is-enough})}
Let $s,t$ be guarded forests with
the same guarded profile. Our goal is to show that
$v\alpha(s)=v\alpha(t)$. 

Let $\tau$ be the guarded profile of $s$ and $t$. Let
$(w_1,a_1),\ldots,(w_m,a_m)$ be all elements of $guard$, and let
$(k_i,x_i)$ be the value $\tau(w_i,a_i)$. Recall that $k_i$ is the
number of times a path of the form $pa_i$ appears in the forest with
$\alpha(p)=w_i$.  By repeatedly applying~\eqref{eq:path-ident},
horizontal commutativity and aperiodicity, we know that the type of
$s$ is 
\begin{eqnarray*}
  \alpha(s)  = \quad  \sum_{i : k_i \ge 1} w_i \alpha(a_i)h_i \quad + \quad \sum_i  (k_i-1)\cdot w_i
  \alpha(a_i)0\ ,
\end{eqnarray*}
for some $h_1,\ldots,h_n \in H$ such that each $h_i$ belongs to the
$v+$-equivalence class $x_i$. Likewise, we can decompose 
\begin{eqnarray*}
  \alpha(t)  = \quad  \sum_{i : k_i \ge 1} w_i \alpha(a_i)g_i \quad + \quad \sum_i  (k_i-1)\cdot w_i
  \alpha(a_i)0\ .
\end{eqnarray*}

So the only difference between the types of $s$ and $t$ is that the
first type uses $h_1,\ldots,h_n$ and the second type uses
$g_1,\ldots,g_n$. However, we know that the types $h_i$ and $g_i$ are
$v+$-equivalent, for any $i$. We will conclude the proof by showing
that
\begin{eqnarray*}
        v(w_i\alpha(a_i)h_i+h)=v(w_i\alpha(a_i)g_i+h)
\end{eqnarray*}
holds for any $h \in H$. By applying the equality above for all
$i=1,\ldots,n$, we get the desired $v\alpha(s)=v\alpha(t)$. By
definition of $v+$-equivalence, the above equality would follow if we
showed that $v(v_i\alpha(a_i)\hole+h) \not \sim v$. This will be shown
in Lemma~\ref{lemma:non-stab}.
\qed

\begin{lem}\label{lemma:non-stab}
  If $vu \not \sim v$ then $v(u+h) \not \sim v$.
\end{lem}
\proof
  Toward a contradiction, assume that $w$ is such that 
  \begin{eqnarray*}
    v(uw+h)  = v
  \end{eqnarray*}
  We assume that $(uw+h)$ is idempotent. By~\eqref{eq:add-side-omega},
  we get
  \begin{eqnarray*}
    v=     v(uw+h)(uw+h + \omega(uw0+h))
  \end{eqnarray*}
  \begin{eqnarray*}
     v(uw+h + \omega(uw0+h)) =  v(uw+\omega(uw0+h))
  \end{eqnarray*}
  Let $x= uw+\omega(uw0+h)$. By the above we know that $vx=v$. By
  definition of $x$ we know that $x=x+h$, and in particular $x^\omega
  = x^\omega +h$. By identity~\eqref{eq:fall-apart}, we get
  \begin{eqnarray*}
      x^\omega=x^\omega (x^\omega +h) = x^\omega(1 + h)
  \end{eqnarray*}
  Therefore,
  \begin{eqnarray*}
    v = vx = vx^\omega = vx^\omega(1  + h) = v(1 +h)
  \end{eqnarray*}
\qed

Now we proceed to prove Lemma~\ref{lemma:profile-can-be-described}.







\proof
We will show that for each 
\begin{eqnarray*}
(w,a) \in guard \quad  \mbox{and} \quad 
(i,x) \in   \set{0,1,\ldots,\omega} \times [H]_{\equiv_{v+}}
\end{eqnarray*}
there is a path language $L_{(w,a),(i,x)}$ such that 
\begin{eqnarray*}
 t \in   L_{(w,a),(i,x)} \qquad \iff \qquad \tau_t(w,a)=(i,x)
\end{eqnarray*}
holds for any guarded forest $t$. This gives 
Lemma~\ref{lemma:profile-can-be-described} by setting
\begin{eqnarray*}
  L_\tau = \bigcap_{(w,a) \in guard} L_{(w,a),\tau(w,a)}\ .
\end{eqnarray*}
Fix $(w,a)$ and $(i,x)$. The easier part is to enforce that the first
coordinate of $\tau_t(w,a)$ is $i$: we just have to say that the
forest $t$ has $i$ paths in the word language 
\begin{eqnarray*}
  K_{w,a}  = \set {a_1 \cdots a_n a \in A^+  : \alpha(a_1 \cdots a_n \hole) =
    w} \ .
\end{eqnarray*}
Only slightly more effort is required in enforcing that the second
coordinate of $\tau_t(w,a)$ is $x$.  By
Lemma~\ref{lemma:induction-assumption}, we know that the set $M_x$ of
forests whose type is in the $v+$-equivalence class $x$ is defined by
a boolean combination of path formulas. To enforce that the second
coordinate of $\tau_t(w,a)$ is $x$, we use the same boolean
combination, except that every word language is prefixed by $K_{w,a}$.
\qed

As we promised before, we now prove that if the path algebra $(H,V)$
in the statement of Theorem~\ref{thm:path-algebra} is vertically
aperiodic, then the path language only needs to use first-order
definable word languages.  It suffices to look at the only place where
we actually wrote word languages: in the lemma above. The word
language $K_{w,a}$ is a word language obtained by concatenating $a$ to
a word language that is recognized by the vertical monoid $V$, via the
morphism $a \mapsto \alpha(a\hole)$.  Since $V$ is aperiodic, we can
use the Sch\"utzenberger and McNaughton-Papert theorem to conclude
that $K_{w,a}$ is first-order definable.

Actually, the argument above can be further generalized to any variety
of word languages given by monoids such that the corresponding
language class is closed under concatenation and contains the one
letter languages $\set{a}$. Note that any such language class
necessarily contains all first-order logic, since it captures all
star-free expressions.



\section{Multicontexts and Confusion}\label{section.multicontexts}

Here we find necessary conditions for a forest algebra to be a
CTL-algebra, an FO-algebra or a graded PDL-algebra.  We use these conditions to show that certain languages cannot be expressed in CTL, FO, or PDL. The conditions we find are essentially the
absence of certain kinds of configurations in the forest algebra,
analogous to the `forbidden patterns' of
Cohen-Perrin-Pin~\cite{perrin-pin} and Wilke~\cite{wilke}.

Let $A$ be a finite alphabet.   A {\it multicontext} $p$ over $A$ is a
forest in which some of the leaves have been replaced by a special
symbol $\hole,$ each occurrence of which is called a {\it hole} of the
multicontext.  A special kind of multicontext, called a {\it uniform}
multicontext, is one in which every leaf node is a hole, and all
subtrees at the same level are identical. For example
$$a(b(c\hole+c\hole))+a(b(c\hole+c\hole))$$
is a uniform multicontext.

The holes are used for substitution. The holes are independent in the sense that different forests can be substituted into different holes. The set of holes of a multicontext $p$ is denoted ${\rm holes}(p).$ A
{\it valuation} on $p$ is a map $\mu:{\rm holes}(p)\to X,$ where $X$
can be a set of forests, or of multicontexts, or elements of $H,$
where $(H,V)$ is a forest algebra.  The resulting value, $p[\mu],$
found by substituting $\mu(x)$ for each hole $x,$ is consequently
either a multicontext, a forest, or an element of $H.$ In the last
case, we are assuming the existence of a homomorphism
$\alpha:A^{\Delta}\to (H,V),$ evaluated at the nodes of $p.$

Given a set $G\incl H$ we write $p[G]$ for the set of all possible
values of $p[\m]$ where $\m:{\rm holes}(p)\to G$.  When $G=\set{g}$ is
a singleton, we just write $p[g]$.  For $g\in G$ and $x\in{\rm holes}(p)$ we
define $p[g/x]$ to be the multicontext
 that results from $p$ by putting a tree that evaluates to $g$ in the hole
$x$.  (In particular, $p[g/x]$ has one less hole than $p$.)

We now define the various type of forbidden patterns for forest algebra.

\subsection{Horizontal confusion}
Let $(H,V)$ be a forest algebra.  As above, we assume the existence of
a homomorphism from $A^{\Delta}$ into $(H,V)$ in order to define the
valuations on $p$ with values in $H$. We say that $(H,V)$ has
\emph{horizontal confusion} with respect to a multicontext $p$ and a
set $G\incl H$ with $|G|>1$ if for every $g\in G$ and $x\in{\rm holes}(p)$:
\begin{equation*}
 G \subseteq p[g/x][G].
\end{equation*}
Intuitively, this means that fixing the value of one of the holes of
$p$ still allows us to obtain any element of $G$ by putting suitable
elements of $G$ into the remaining holes. 

\subsection{$k$-ary horizontal confusion}
We can define a stronger version of confusion, which seems to be satisfied by fewer forest algebras. In the stronger version, we are allowed to fix the value in not just one, but in  $k \ge 1$ holes: We say that the forest algebra $(H,V)$ has \emph{$k$-ary horizontal confusion} with respect to a multicontext $p$ and a set $G\incl H,$ with $|G|>1,$ if for all $g_1,\ldots, g_k\in G$ and $x_1,\ldots,x_k\in{\rm holes}(p),$
\begin{equation*}
 G =  p[g_1/x_1,\cdots,g_k/x_k][G].
\end{equation*}

The following lemma shows that the stronger notion is in fact equivalent to horizontal confusion, because we can always amplify horizontal confusion to $k$-ary horizontal confusion for arbitrary $k.$

\begin{lem}
\label{lem:amplifying-confusion}
Suppose $(H,V)$ has horizontal confusion with respect to a multicontext $p$ and a subset $G$ of $H,$ with underlying homomorphism $\phi:A^{\Delta}\to (H,V).$ Let $k>0.$  Then there is a multicontext $p_k$ such that $(H,V)$ has $k$-ary horizontal confusion with respect to $p_k,$  $G$ and $\phi.$
\end{lem}

\proof
We prove this by induction on $k.$  We have $p_1=p,$ by hypothesis.  If $k>1,$ we define $p_k$ by placing a copy of $p_{k-1}$ in each of the holes of $p.$  To see that this works, fix the values in $G$ of $k$ of the holes holes of $p_k.$  If the $k$ holes do not all belong to the same copy of $p_{k-1},$ then each copy has fewer than $k-1$ holes fixed, and thus we can set the values in the remaining holes to get any elements of $G$ we want in the holes of $p,$ and consequently any element of $G$ as a value of $p_k.$  If the $k$ holes all belong to the same copy of $p_{k-1},$ then the resulting value $g\in G$ produced by this copy might be determined, but this will only constrain the value in one of the holes of $p.$  Since $p$ has horizontal confusion, we can set the remaining holes of $p$ to values $g_1,\ldots, g_r$ to obtain any desired value as output, and we can in turn set the values of the other copies of $p_{k-1}$ to obtain these values $g_1,\ldots, g_r.$
\qed

\subsection{Vertical confusion}
We say  that the forest algebra $(H,V)$ has \emph{vertical  confusion} with respect to
a multicontext $p$ and a set  $\set{g_0,\dots,g_{k-1}}\incl H$ with $k>1$
if for every $i=0,\dots,k-1$:
\begin{quote}
  $p[g_i]=g_j$ where $j=(i+1)\pmod k$.
\end{quote}
This condition is weaker than periodicity of vertical monoid, because
$p$ is a multicontext, and not just a context. For instance, consider the syntactic forest algebra of the tree language $L$, which consists of trees where every node has two or zero children, and where every leaf is at even depth. 

\subsection{Confusion Theorem}
The next theorem shows how the various types of confusion are forbidden in CTL-, FO- and PDL-algebras.

\medskip
\begin{thm}[Confusion Theorem]~\label{thm:confusion}
\begin{iteMize}{$\bullet$}

\item If $(H,V)$ is a CTL-algebra, it does not have  vertical confusion  with respect to any multicontext. 

\item If $(H,V)$ is an FO-algebra, it does not have  vertical confusion with respect to any uniform multicontext.

\item If $(H,V)$ is a  graded PDL-algebra, it does not have  horizontal
  confusion with respect to any multicontext.
\end{iteMize}
\end{thm}


\proof
For each of the three kinds of confusion and each of the corresponding language classes, we will show that the nonconfusing property  {\it (a)} holds for the elements of the algebraic base of the class, {\it (b)} is preserved by wreath products, and  {\it (c)} is preserved by quotients and subalgebras.

We begin with vertical confusion and the class CTL which has $\mathcal{U}_2$ as an algebraic base.  Let $\phi:A^{\Delta}\to \mathcal{U}_2$ be a homomorphism, and suppose $p$ is a multicontext over $A$ such that $U_2$ has vertical confusion with respect to $p$ and $\phi.$    Since $\mathcal{U}_2$ is distributive, we have
$$p[g]=\sum_{u\in\pi}\phi(u)g +\sum_{v\in\rho}\phi(v)\cdot 0,$$
where the first sum ranges over the set $\pi$ of paths in $p$ from a root to the parent of a hole, and the second over the set $\rho$ of paths from the root to a leaf.  We claim that for $g\in\{0,\infty\},$  $p[p[g]]=p[g].$  This follows easily from an enumeration of the possible cases:  If $g=p[g],$ then the claim is trivial, so we can assume that either $g=\infty$ and $p[g]=0,$ or $g=0$ and $p[g]=\infty.$  In the first case, every path in $\pi$ has a prefix $wa$ with $a\in A,$ $\phi(a)=c_0,$ and $\phi(b)=1$ for every letter $b$ of $w,$ and every path in $\rho$ has either this form or has $\phi(b)=1$ for every letter $b.$  It follows that $p[0]=0.$  In the second case, some path in $p$ has a prefix $wa$ with $\phi(a)=c_{\infty}$ and $\phi(b)=1$ for every letter $b$ of $w,$ and thus $p[\infty]=\infty.$  Since $p[p[g]]=p[g]$ for all $g$ in the horizontal monoid of ${\mathcal U}_2,$ we cannot have vertical confusion.

We now consider the base algebras for $FO[\prec].$ Suppose that
$\phi:A^{\Delta}\to (H,V)$ is a homomorphism into an aperiodic path
algebra.  Then, by Theorem~\ref{thm:path-algebra}, every
language recognized by $\phi$ is an fo path language---that is, a
boolean combination of languages of the form $E^kL,$ where $L\incl
A^*$ is a first-order definable word language. Let $p$ be a uniform
multicontext over $A.$ Since $p$ is uniform, every maximal path in $p$
has the same label $u\in A^*.$ We can dispense with the case where $p$
has a single hole, because then $p[g]$ reduces to $\phi(u)\cdot g,$
and by aperiodicity of the vertical monoid we have, for some $n\ge 0,$
$p^{n+1}[g]=\phi(u^{n+1})g=\phi(u^n)g=p^n[g],$ so there is no vertical
confusion. We thus suppose that $p$ has at least two holes, so that
$p^n$ is a multicontext with at least $2^n$ holes. Since every
language recognized by $\phi$ is an fo path language, there exists a
congruence $\sim$ of finite index on $A^*$ and an integer $k>0$ such
that $A^*/\sim$ is aperiodic, with the following property: If $s,t\in
H_A$ are such that for every $\sim$-class $\kappa,$ the number of
paths from the root of $s$ in $\kappa$ is equal, up to threshold $k,$
to the number of paths from the root of $t$ in $\kappa,$ then
$\phi(s)=\phi(t).$ (`Equal up to threshold $k$' means
either equal, or both at least $k.$) Since $A^*/\sim$ is aperiodic,
there is an integer $r$ such that $u^r\sim u^{r+1}.$ Let $g\in H,$ and
let $s$ be any forest such that $\phi(s)=g.$ Choose $q$
such that both $q>r$ and $2^q>k.$ Now consider the forests
$p^{q+1}[s]$ and $p^{q+2}[s].$ Suppose that a word occurs as the label
of a path from the root in $p^{q+2}[s]$ more times than it does in
$p^{q+1}[s].$ Then, since $p$ is uniform, the word must have the form $u^{q+1}v,$ and since
$u^{q+1}v\sim u^qv,$ a word in the same $\sim$-class occurs at least
$2^q>k$ times in $p^{q+1}[s].$ It follows that $p^{q+1}[g]=
\phi(p^{q+1}[s])=\phi(p^{q+2}[s])=p^{q+2}[g],$ so there is no vertical
confusion.

We now consider the base algebras for graded PDL, so we suppose $\phi:A^{\Delta}\to(H,V)$ is a homomorphism onto a path algebra $(H,V)$ which has horizontal confusion with respect to a multicontext $p$ and a set $G\incl H,$ with $|G|>1.$  Let $G=\{g_1,\ldots, g_n\},$ and let $s_1,\ldots, s_n$ be forests such that $\phi(s_i)=g_i$ for all $1\le i\le n.$ As above, there is a congruence $\sim$ of finite index on $A^*$ and an integer $k,$ such that if two forests agree on the number of paths threshold $k$ and modulo $\sim,$ then they have the same image under $\phi.$ Let $m$ be the index of $\sim.$ (The only difference from the previous case is that we no longer have $A^*/\sim$ aperiodic.) By Lemma~\ref{lem:amplifying-confusion}, there is a context $q$ such that $(H,V)$ has $km$-ary horizontal confusion with respect to $q.$  We order the classes of $\sim$ arbitrarily as $\kappa_1,\ldots, \kappa_m.$ We proceed to insert forests from $s_1,\ldots, s_n$ into the holes of $q$ according to the following algorithm:  For each $\kappa_i$ in turn, we ask if there is a way to substitute copies of the $s_j$ into the holes we have not yet filled in order to obtain at least $k$ paths in $\kappa_i.$  If so, we perform the necessary insertions; if not we insert enough copies of the $s_j$ to obtain the maximum possible number of paths in $\kappa_i.$  At the end of the process, we will have filled no more than $km$ holes.  However, no further substitution of forests $s_j$ for the remaining holes can increase the number, threshold $k$ of paths in any class of $\sim,$ and thus no matter how we fill the remaining holes, the value under $\phi$ will be the same.  But because of the $km$-ary confusion, we should be able to obtain any value in $G$ by appropriately filling the remaining holes.  Thus $|G|=1,$ so there is no horizontal confusion.

We now show closure under wreath product.  Suppose first that neither $(H_1,V_1)$ nor $(H_2,V_2)$ has vertical confusion with respect to any multicontext. Let $\gamma$ be a homomorphism from $A^{\Delta}$ into the wreath product $(H,V)=(H_1,V_1)\circ (H_2,V_2).$ Suppose $(H,V)$ has vertical confusion with respect to some multicontext $p$ with underlying homomorphism $\gamma.$  There thus exist $g_i=(g_i^{(1)},g_i^{(2)})\in H=H_1\times H_2,$ with $i=0,\ldots, n-1,$ such that $p_{\gamma}[g_i]=p_{\gamma}[g_{(i+1)\bmod n}]$ for $0\le i < n.$  (Note that here we explicitly indicate the  homomorphism $\gamma,$  since we will be shortly be applying the multicontext $p$ with respect to other homomorphisms.) By Theorem~\ref{thm:sequential-composition}, $\gamma=\alpha\otimes\beta,$ where $\alpha:A^{\Delta}\to(H_1,V_1)$ and $\beta:(A\times H_1)^{\Delta}\to(H_2,V_2)$ are homomorphisms.  When we project onto the left co-ordinate, we obtain 
$$p_{\alpha}[g_i^{(1)}]=g_{(i+1)\bmod n}^{(1)}.$$
Since $(H_1,V_1)$ does not have vertical confusion, all the $g_i^{(1)}$ must be equal.  We will denote their common value by $g^{(1)}.$   We now form a new multicontext $p^{(\alpha,g^{(1)})}$ by first substituting any forest  evaluating to $g^{(1)}$ for the holes in $p,$ which gives a forest $t,$ then forming the forest $t^{\alpha},$ and finally restoring the original holes.  The resulting multicontext has the same shape as $p,$ but its nodes are now labeled by elements of $A\times H_1.$  Because the value $g^{(1)}$ is stable after each application of $p_{\alpha},$ we find that
$p_{\beta}^{(\alpha,g^{(1)})}[g_i^{(2)}]$ is identical to the right-hand coordinate of
$p_{\gamma}(g_i),$ and thus we have 
$$p_{\beta}^{(\alpha,g^{(1)})}[g_i^{(2)}]=g_{(i+1)\bmod n}^{(2)}$$
for all $0\le i < n.$  Since $(H_2,V_2)$ does not have vertical confusion, we find that all the $g_i^{(2)},$ and consequently all the $g_i,$ are identical.  So $(H,V)$ does not have vertical confusion.

In the case of vertical confusion with respect to uniform multicontexts, the proof is the same; we simply note that the multicontext $p^{(\alpha,g^{(1)})}$ defined above is uniform whenever $p$ is.  In the case of horizontal confusion with respect to some $G\subseteq H_1\times H_2,$ we use essentially the same argument:  absence of confusion in the left coordinate permits us to reduce $G$ to a set of the form $\{g^{(1)}\}\times G_2,$ and we find that $H_2$ has horizontal confusion with respect to $p^{(\alpha,g^{(1)})}$ and $G_2,$ so that $|G_2|=1,$ and hence $|G|=1.$

We now show that in each case the non-confusing property is preserved under division.  For subalgebras, this is trivial, but for quotients, there is something to prove.  Accordingly, suppose that $\psi:(H_1,V_1)\to (H_2,V_2)$ is a surjective homomorphism of forest algebras.  Let $\phi:A^{\Delta}\to (H_2,V_2)$ be a homomorphism. We can lift this to a homomorphism $\pi:A^{\Delta}\to (H_1,V_1)$ such that $\psi\pi=\phi.$ First suppose $(H_2,V_2)$ has vertical confusion with respect to some multicontext $p$ and $\phi.$  We will show $(H_1,V_1)$ has vertical confusion with respect to $p$ and $\pi.$  Vertical confusion in $(H_2,V_2)$ gives us a sequence $g_0,\ldots, g_{n-1}$ of elements of $H_2$ with $n>1$ such that $p_{\phi}[g_i]=g_{(i+1)\bmod n}$ for all $0\le i< n.$  Choose an element $h_0\in H_1$ such that  $\psi(h_0)=g_0,$ and define $h_1,h_2,\ldots$ by $h_{i+1}=p_{\pi}[h_i].$  By finiteness, there exist $j<k$ such that $h_j$ = $h_k.$  Since $\psi(h_j)=g_{j\bmod n}$ and $\psi(h_k)=g_{k\bmod n},$ we have $k-j$ is a multiple of $n,$ and in particular, $k-j>1.$  We thus have
$$p_{\pi}[h_{j+i}]= h_{j+(i+1)\bmod (k-j)},$$
which gives vertical confusion in $(H_1,V_1).$   

Now suppose that we have  horizontal confusion in $(H_2,V_2)$ with respect to $\phi.$ We will show how to obtain  horizontal confusion in $(H_1,V_1).$ Let $m=|H_1|.$  By Lemma~\ref{lem:amplifying-confusion}, there is a multicontext $p$ such that $(H_2,V_2)$ has $m$-ary horizontal confusion with respect to  $\phi$ and some set $G'\incl H_2.$ Let $G=\psi^{-1}(G').$  For $k>0,$ set $G^k_1=p^k[G]$.  Since $p[G']=G',$ we have $\psi(G^k_1)=G'$ for all $k.$  In particular, $G^1_1\subseteq G,$ and by repeatedly applying $p$ to both sides of this inclusion we obtain $G^{k+1}_1\subseteq G^k_1$ for all $k.$  Thus this sequence eventually stabilizes, so we have some $n$ for which $p[G^n_1]=G^n_1.$  Let us set $G_1=G^n_1.$  Now it may be that $(H_1,V_1)$ has horizontal confusion with respect to $p,\pi,$ and $G_1.$  If not, there is some hole $x$ of $p$ and $g_1\in G_1$ such that $p[g_1/x][G_1]\subsetneq G_1.$  So we let $p'$ be the multicontext that results from substituting a forest that evaluates under $\pi$ to $g_1$ for $x,$ and set $G_2^1=p'[G_1].$  Note that $(H_2,V_2)$ has $(m-1)$-ary horizontal confusion with respect to $p'$ and $\phi,$ so we still have $\psi(G_2^1)=G',$  as well as $G_2^1=p'[G_1]=\subsetneq G_1,$ so that $|G_2^1|<|G|.$  We now repeat the procedure above, applying $p'$ to $G_2^1$ until the sequence stabilizes at a set $G_2,$ then checking if the result is a horizontal confusion for $(H_1,V_1),$ and filling a hole of $p'$ if it is not.  We have
$$|G'|\le \cdots |G_k| < |G_{k-1}<\cdots |G_1|\le |G|,$$
so the process will terminate after no more than $|G|-|G'|$ generations, giving a horizontal confusion in $(H_1,V_1).$
\qed

\begin{thm}
  It is decidable if a given forest algebra has  horizontal confusion,  vertical
  confusion, or vertical confusion with respect to a uniform context.
\end{thm}
\begin{proof}
  Confusion in a forest algebra $(H, V)$ appears to depend on the
  choice of alphabet $A$, a multicontext $p$ over $A$, and a morphism
  from $A^\Delta$ into $(H, V)$. Observe, however, that we can
  restrict attention to a single alphabet and morphism: Consider $V$
  as a finite alphabet, and the morphism $\beta : V^\Delta \to (H,V)$
  induced by the identity map on $V$. If $(H,V)$ has a confusion with
  respect to a multicontext $p$ over $A$ and morphism $\alpha :
  A^\Delta \to (H, V)$, then we can transform it into a confusion of
  the same type with respect to $V$ and $\beta$ in the obvious
  fashion, replacing each node label $a\in A$ of $p$ labeled by
  $\alpha(a) \in V$. Thus in the argument below, we suppress explicit
  mention of an alphabet and morphism and work simply with the
  elements of $V$.

%
%

\paragraph*{\it Vertical confusion.} Testing whether $(H,V)$ has vertical confusion with respect to some multicontext reduces to verifying whether a certain monoid containing $V$ is aperiodic.  If $v, w\in V,$ we define $v+w$ to be the transformation on $H$ given by
$$(v+w)h=vh+wh$$
for all $h\in H.$
Let $\hat V$ be the collection of all maps on $H$ containing $V$ and closed under composition and addition.  $\hat V$ then consists of all multicontexts over $(H,V).$  Furthermore, $\hat V$ is effectively computable from $V,$ since whenever we have a set $U$ of transformations on $H,$ we can check for each $v,w\in U$ whether $v+w$ and $vw$ belong to $U,$ and if not, adjoin them to $U.$ Since there are only finitely many transformations on $H,$ we eventually reach a stage at which we can add no new elements to $U,$ at which point the algorithm terminates.

$\hat V$ is a monoid under composition, and $(H,V)$ is free of vertical confusion if and only if this monoid is aperiodic; {\it i.e.,} if and only if $p^k=p^{k+1}$ for all $p\in \hat V$ and sufficiently large $k,$ which we can determine effectively.

\paragraph*{\it Vertical confusion with respect to a uniform multicontext.} The argument is the same as above, however now we must build a monoid containing $V$ that consists of exactly all the uniform multicontexts.  We accordingly close $V$ under composition and the operations
$$v\mapsto v+v+\cdots +v.$$
Observe that the number of summands in this expression can be bounded above by the size of $H,$ so we can compute this closure effectively as well.  Let us denote the resulting monoid $\tilde{V}.$ $(H,V)$ does not have vertical confusion with respect to any uniform multicontext if and only if $\tilde{V}$ is aperiodic.

\paragraph*{\it Horizontal confusion.} We now test if $(H,V)$ has horizontal confusion. The algorithm first guesses the set $G$. For a multicontext $p$, we define its profile to be the set 
	\begin{align*}
		\pi(p) = \set{p[g/x][G] : g \in G, x \in {\rm holes}(p)} \times p[G]\qquad  \in P(P(H)) \times P(H).
	\end{align*}
	The forest algebra  has horizontal confusion with respect to a multicontext $p$ and $G$ if and only if the profile $\pi(p)$ only has supersets of $G$ on the first coordinate. Therefore, to determine if the forest algebra has horizontal confusion, it suffices to compute the set 
	\begin{align*}
	Y = \set{ \pi(p) : \mbox{$p$ is a multicontext}}.
	\end{align*}
This set is computed using a fix-point algorithm, since it is the least set that satisfies the properties listed below. (In the implications, we lift the forest algebra operations to sets $F \subseteq H$ and families of sets $\Ff \subseteq P(H)$ in the natural way.)
\begin{eqnarray*}
	(\set{\set{g} : g \in G}, G) \in Y\\
	(\Ff,F)\in Y \Rightarrow (v\Ff,vF) \in Y & & \mbox{ for every $v \in V$}\\
	(\Ff_1,F_1),(\Ff_2,F_2) \in Y \Rightarrow (\Ff_1+F_2 \cup F_1+\Ff_2, F_1+F_2)
\end{eqnarray*}
	\end{proof}

\section{Applications}\label{section.applications}

Here we apply the results of the preceding section to exhibit a forest language in CTL* that is not in CTL, a language in PDL that is not in $FO[\prec],$ and a language that is not in graded PDL.   All of our examples have syntactic forest algebras with aperiodic vertical monoids, and all the classes in question contain languages with arbitrarily complicated aperiodic vertical monoids, so we really do need  machinery of forest algebras to give algebraic proofs of these separations.

\subsection{ Forests with a maximal path in $(ab)^*$ }
Consider the set $L_1$ of forests over $A=\{a,b\}$ in which there is a maximal path---that is, a path from a root to a leaf--- in $(ab)^*.$  This language is in CTL*.  To see this, note that $\phi= \tle (A^+)$ is a forest formula in CTL* defining the set of nonempty forests. Consider the formally disjoint formulas
\begin{eqnarray*}
	\phi_1 = b \land \phi \qquad \phi_2 = b \land \neg \phi_1 \qquad \phi_3 = \neg \phi_1 \land \neg \phi_2.
\end{eqnarray*}
The formula $\phi_1$ holds in non-leaf nodes with label $b$, the formula $\phi_2$ holds in leaves with label $b$, and the formula $\phi_3$ holds in nodes with label $a$. Then $L_1$ is defined by the CTL* forest formula $\tle((\phi_3\phi_1)^*(\phi_3\phi_2)).$  We claim that $L_1$ is not in CTL.
To do this, by Theorem~\ref{thm:confusion}, we need only exhibit a multicontext $p$ with respect to which the syntactic forest algebra of $L_1$ has vertical confusion. Let $p=a\hole + b\hole.$  Let $h_0$ be the class of the tree $b$ in the syntactic congruence of $L_1,$ and let $h_1$ be the class of the tree $ab.$  Observe that $h_0$ and $h_1$ are distinct horizontal elements of the syntactic algebra, since $h_1$ contains elements of $L_1$ and $h_0$ does not.  We have vertical confusion, because $p[h_0]$ is then the class of $ab+bb,$ which is $h_1,$ and $p[h_1]$ is the class of $aab+bab,$ which is $h_0.$  
 \subsection{Binary trees with even path length}
  This example uses  {\it unlabeled binary trees}, which are 
  trees over a one-letter alphabet $\{a\}$ where every node has zero or two children. 
Let $L_2$ be the set of unlabeled binary trees where  every path from the
  root to a leaf has even length.    Let $p$ be the uniform multicontext $a(\hole+\hole).$ Let $h_0$ denote the set  of binary trees in which every maximal path has even length, and $h_1$ the set of binary trees in which every maximal path has odd length. These are distinct classes in the syntactic congruence of $L_2.$  Obviously $p[h_0]=h_1$ and $p[h_1]=h_0,$ so we have vertical confusion with respect to a uniform multicontext, and thus by  Theorem~\ref{thm:confusion}, $L_2$ is not in $FO[\prec].$

 An argument due to Potthoff~\cite{potthoff} can be used to show that $L_2$ is definable in first-order logic in which there is both the ancestor and the next-sibling relations. Languages definable in $FO[\prec]$ are obviously in the intersection
  of the class of languages definable in $FO$ with $\prec$ and the
  next-sibling relationship, and the class of languages $L$ with
  commutative $H_L.$ This example shows that the containment is
  strict. Note that $L_2$ is expressible in graded PDL so we have also
  established that the languages in graded PDL with aperiodic forest
  algebras need not be definable in $FO[\prec]$ (there is even an
  example, also due to Potthoff, which shows that languages definable
  in graded PDL with aperiodic forest algebras need not be definable
  in $FO$ with $\prec$ and the next-sibling relationship).

\subsection {\it (Boolean expressions).} Consider the set $L_3$ of
  trees over the alphabet $\{0,1,\vee,\wedge\}$ that are well-formed
  boolean expressions ({\it i.e.}, all the leaf nodes are labeled 0 or
  1, and all the interior nodes are labeled $\vee$ or $\wedge$) that
  evaluate to 1. $L_3$ is contained in a single equivalence class of
  the syntactic congruence, as is the set of well-formed trees that
  evaluate to 0.  We denote the corresponding elements of $H_{L_3}$ by
  $h_1$ and $h_0.$

  Now consider the  multicontext
  $p=\vee(\wedge(\hole+\hole)+\wedge(\hole+\hole)).$ We can fix a value $1$ or $0$ in any single hole, and then set the remaining holes to obtain either a tree evaluating to 1 or a tree evaluating to 0.  Thus the syntactic algebra of $L_3$
  has horizontal confusion with respect to the multicontext $p$ and the set $\{h_0, h_1\},$ and so is not in graded PDL. Observe that the vertical component of the syntactic
  algebra of $L_3$ is aperiodic: In contrast to the word case,
  languages recognized by aperiodic algebras are not necessarily
  expressible in first-order logic, or even in graded PDL.

\subsection{Horizontally idempotent and commutative algebras}

Obviously, we can separate CTL$^*$ and PDL from $FO[\prec]$ and graded
PDL, respectively, because the syntactic algebras for the former
classes have idempotent and commutative horizontal parts, while for
the latter the horizontal components need only be aperiodic and
commutative. Thus, for example, any language in $FO[\prec]$ that fails
to satisfy the idempotency condition is not in CTL$^*$. We can use our
algebraic methods to show that this is in fact the {\it only}
distinction:


\bigskip
\begin{thm}~\label{thm:slwreath} Let $(H,V),$ $(H_j,V_j),$ $j=1,\ldots,k$ be forest
  algebras such that $H$ is idempotent and commutative, each
  $(H_i,V_i)$ is a path algebra, and such that $(H,V)$ divides
  $(H_1,V_1)\circ\cdots\circ(H_k,V_k).$ Then each $(H_i,V_i)$ has a
  distributive homomorphic image $(H_i',V_i')$ such that $(H,V)$
  divides $(H_1',V_1')\circ\cdots\circ (H_k',V_k').$

\end{thm}

\proof
Let $(H,V)$ be a path algebra.  We define
$e(H)$ to be the set of idempotents of $H.$  By the commutativity of $H,$ the sum of two idempotents is idempotent.  Thus $e(H)$ is an idempotent and commutative submonoid of $H.$  If $h\in H$ and $k$ is a nonnegative integer, we denote by $k\cdot h$ the sum of $k$ copies of $h.$  We also denote by $\omega h$ the unique idempotent in $\{k\cdot h:k\in  \Nat\}.$ Since $H$ is aperiodic and commutative, there exists $k$ such that $\omega\cdot h=k\cdot h=(k+1)\cdot h$ for all $h\in H.$

For every  $v\in V,$ we define a function $\bar v : e(H) \to e(H)$ by

$$\bar{v}\cdot e=\omega(ve).$$  
We define a forest algebra $(e(H),\bar V)$ as follows. The horizontal monoid is $e(H)$. The vertical monoid is $\bar V = \set{\bar v : v \in V}$, with function composition. The action is by applying the function $\bar v$ to an argument $e \in e(H)$. To prove that this is a forest algebra, we need to show that for any element $e \in e(H)$, there is an element $\bar v \in \bar V$ such that $\bar v f = e +f $ holds for any $f \in e(H)$. This element is simply $\bar{e + \hole}$. Indeed,
$$(\overline{e + \hole})\cdot f=\omega ((e + \hole)f)=\omega (e+f)=e+f.$$
This concludes the proof that $(e(H),\bar V)$ is a forest algebra.

We now show that $(e(H),\bar V)$   is distributive. In other words, the following identity holds for any $v \in V$ and $e_1,e_2 \in e(H)$.
\begin{equation}\label{eq:sldist}{\bar v}(e_1+e_2)={\bar v}e_1+\bar{v}e_2\end{equation}
	Using the first identity in the definition of path algebras~\eqref{eq:path-ident}, we obtain
	\begin{eqnarray*}
	\bar{ v}(e_1+e_2)&=&\omega  v(e_1+e_2)\\
	                                    &=&\omega (v(e_1+e_2)+v(e_1+e_2))\\
	                                    &=&\omega (v(e_1+e_2+e_1+e_2)+v\cdot 0)\\
	                                    &=&\omega  (v(e_1+e_2)+v\cdot 0)\\
	                                    &=&\omega (ve_1+ve_2)\\
	                                    &=&\omega  ve_1+\omega\cdot ve_2\\
	                                    &=&{\bar v}e_1+\bar{v}e_2.
	\end{eqnarray*}

Finally, we prove that  the function
\begin{align*}
	\alpha(h)= \omega \cdot h \qquad \alpha(v)=\bar v
\end{align*}
is a forest algebra homomorphism
\begin{align*}
	\alpha : (H,V) \to (e(H),\bar V).
\end{align*}

Clearly $\alpha$ preserves $+$. It remains to show that it
preserves the remaining two operations of forest algebra, namely inserting a forest into a context and composition of two contexts. For inserting a forest into a context, we have
\begin{align*}
	\alpha(vh) = \omega \cdot vh = \bar v h =  \alpha(v)  \alpha(h).
\end{align*}
For composition of two contexts we need to show $\alpha(v_1) \alpha(v_2) = \alpha(v_1 v_2)$. Since $\bar V$ is defined as a set of functions on $e(H)$, we need to show that both sides of the equality describe the same function on $e(H)$. In other words, we have to prove that for every $e \in e(H)$,
\begin{equation}\label{eq:slmult}\bar{v_1}(\bar{v_2}e)=\overline{v_1v_2}e\end{equation}
First note that the path algebras property~\eqref{eq:path-ident} implies that for all $h\in H,$ $v\in V,$ $m\in \set{1,2,\ldots},$ we have
$$m\cdot(vh)=v(m\cdot h)+(m-1)\cdot(v0).$$
Thus by aperiodicity of $H,$
$$\omega(vh)=v(\omega h)+\omega(v 0).$$
If $e\in H$ is idempotent, this becomes
$$\omega  (ve)=ve+\omega( v0).$$  Consequently we have

\begin{eqnarray*}
\overline{v_1v_2}e &=& \omega (v_1v_2e)\\
                                   &=& v_1\omega  (v_2e)+\omega  (v_1\cdot 0)\\
                                   &=&  v_1(v_2e+\omega (v_2\cdot 0))+\omega (v_1\cdot 0)\\
                                   &=&  v_1(v_2e+\omega (v_2\cdot 0)+\omega (v_2\cdot 0))+\omega (v_1\cdot 0)\\
                                   &=& v_1(\omega (v_2e)+\omega (v_2\cdot 0))+\omega (v_1\cdot 0)\\
                                   &=& v_1\omega (v_2e+\omega (v_2\cdot 0))+\omega (v_1\cdot 0)\\
                                   &=& \omega  v_1(v_2e+\omega (v_2\cdot 0))\\
                                   &=&\omega v_1(\omega v_2e)\\
                                   &=& {\bar v_1}({\bar v_2}e).
\end{eqnarray*}

Summing up: We have defined a forest algebra homomorphism
\begin{align*}
	\alpha : (H,V) \to (e(H),\bar V)
\end{align*}
where the target forest algebra is distributive and horizontally commutative and idempotent.

Suppose now that $(H,V)$ is a forest algebra with idempotent and commutative $H$ that divides a wreath product
$$(H_1,V_1)\circ\cdots\circ(H_k,V_k),$$
where each $(H_i,V_i)$ is a path algebra.  To complete the proof of the theorem, we will show that $(H,V)$ divides
$$(e(H_1),\bar{V_1})\circ\cdots\circ (e(H_k),\bar{V_k}).$$
%

We now apply Lemma~\ref{lemma.division} on the equivalence of the two definitions of division. The hypothesis is then that there is a submonoid $H'$ of $H_1\times\cdots\times H_k,$  and a homomorphism $f$ from $H'$ onto $H$ with the following property:  For each $v\in V$ there is $\hat v$ in the vertical monoid of $(H_1,V_1)\circ\cdots\circ(H_k,V_k)$ such that for all $h\in H',$
$$f({\hat v}h)=vf(h).$$
Note that $h$ has the form $(h_1,\ldots,h_k),$ where $h_i\in H_i$ for $i=1,\ldots,k,$ and 
$${\hat v} = (u,g_2,\ldots,g_k),$$
 where $u\in V_1$ and each
$$g_j:H_1\times\cdots\times H_{j-1}\to V_j$$
is a map.  Since $H$ is idempotent, $f(\omega h)=f(h)$ for all $h\in H'.$  We consider the restriction of $f$ to $e(H'),$ which is a subset of $e(H_1)\times\cdots\times e(H_k).$  We will show that for each $v\in V$ there is an element $\tilde v$ of the vertical monoid of $(e(H_1),\bar{V_1})\circ\cdots\circ (e(H_k),\bar{V_k})$ such that 
for all $e\in e(H),$
$$f({\tilde v}e)=v f(e).$$
To do this, we simply alter $\hat v=(u,g_2,\ldots,g_k)$ in the obvious fashion:
$${\tilde  v} = (\bar{u},\bar{g_2},\ldots,\bar{g_k}),$$
where by definition 
$${\bar g_j}(x)=\overline{g_j(x)}.$$

Set $e=(e_1,\ldots,e_k)\in e(H').$ 
We have
\begin{eqnarray*}   
f({\tilde v}e) & =& f ( {\bar u}e_1,\overline{g_2(e_1)}e_2,\ldots,\overline{g_k(e_1,\ldots,e_{k-1})}e_k  )\\
&=& f ( {\bar u}e_1,\overline{g_2(e_1)}e_2,\ldots,\overline{g_k(e_1,\ldots,e_{k-1})}e_k  )\\
			&=& f (\omega(ue_1),\omega g_2(e_1)e_2 ,\ldots,\omega g_k(e_1,\ldots,e_{k-1})e_k) \\
			&=& f(ue_1,g_2(e_1)e_2,\ldots,  g_k(e_1,\ldots,e_{k-1})e_k)\\
			&=&f({\hat v}e)\\
			&=&vf(e),
\end{eqnarray*}
which completes the proof.

\qed

Theorem~\ref{thm.main} immediately yields the following corollary:

\begin{thm}\label{thm.foctlstar} A forest language is definable in CTL$^*$ (respectively
  PDL) if and only if it is definable in $FO[\prec]$ (respectively
  graded PDL) and its syntactic algebra is horizontally idempotent.
\end{thm}

The first of these facts follows from a result of  Moller and
Rabinovich~\cite{moller-rabinovich2} who show that over infinite trees 
properties expressible in CTL$^*$ are exactly the bisimulation-invariant properties expressible in monadic path logic.



\section{Conclusion and further research}

Results like those in Section~\ref{section.applications} are typically proved by model-theoretic methods. Here we have demonstrated a fruitful and fundamentally new way, based on algebra, to study the expressive power of these logics.  

Of course, the big question left unanswered is whether we can
establish effective necessary and sufficient conditions for membership
in any of these classes.  We do not expect that the conditions
established in Theorem~\ref{thm:confusion} are sufficient.  The approach outlined
in Section~\ref{section.ef} may constitute a model for how to proceed:
a deeper understanding of the ideal structure of forest algebras can
lead to new wreath product decomposition theorems.

In a sense, we are searching for the right generalization of
aperiodicity.  For regular languages of words, aperiodicity of the
syntactic monoid, expressibility in first-order logic with linear
ordering, expressibility in linear temporal logic, and recognizability
by an iterated wreath product of copies of the aperiodic unit $U_2$
are all equivalent.  For forest algebras, the obvious analogues are,
respectively, aperiodicity of the vertical component of the syntactic
algebra, expressibility in $FO[\prec],$ expressibility in {CTL}, and
recognizability by an iterated wreath product of copies of
$\mathcal{U}_2.$ As we have seen, only the last two
coincide. Understanding the precise relationship among these different
formulations of aperiodicity for forest algebras is an important goal
of this research.

Another way of looking at this research is that it sets the scene for
a Krohn-Rhodes theorem for trees. The Krohn-Rhodes theorem states that
every transition monoid divides an iterated wreath product of transition monoids
which are either $U _2$ or groups that divide the original
monoid. The ingredients of the theorem are therefore: a notion of
wreath product, a notion of an easy transition monoid $U_2$,
and a notion of a difficult transition monoid (a group). For our purposes here,
we are particularly interested in the (already quite difficult) version of the theorem
which states that every aperiodic transition monoid divides a wreath product of copies of $U_2.$
In this
paper, we have provided some of the ingredients: the wreath product
and the easy objects.  (There are several candidates for the easy
objects, e.g.~simply $\mathcal{U}_2$ or maybe path algebras. 
There are probably several Krohn-Rhodes theorems). We have
provided examples of properties one expects from the difficult objects
(the various types of confusion), but we still have no clear idea what
they are (in other words, what is a tree group?). We have also shown
that the wreath product is strongly related to logics and composition.  Finding (at least one) Krohn-Rhodes
theorem for trees is probably the most ambitious goal of this
research.



\bibliographystyle{plain}
\bibliography{wreathproducts}

\appendix

\end{document}